\documentclass[aps,prl,twocolumn,amsmath,amssymb,gallery,footinbib,superscriptaddress,floatfix]{revtex4-2}
\usepackage{graphicx}
\usepackage{dcolumn}
\usepackage{amsthm}
\usepackage{bm}
\usepackage{hyperref}
\usepackage{tikz}
\usepackage{stackrel,amssymb}
\hypersetup{colorlinks=true, linkcolor=blue, citecolor=blue, urlcolor=magenta, pdftitle={On the search of displacive chiral phase transitions}, pdfauthor={F. Gomez-Ortiz and E. Bousquet}}
\usepackage{fixltx2e}
\usepackage{bbding}
\usepackage{xcolor} 
\usepackage[normalem]{ulem}
\newtheorem{theorem}{Theorem}

\usepackage{tabularx}
\usepackage{multirow}
\usepackage{colortbl}
\usepackage{soul}
\usepackage{cancel}
\usepackage{textcomp}
\usepackage{gensymb}

\begin{document}
\preprint{APS/123-QED}
\title{Chirality Cannot Be Ferroic in Enantiomorphic Space-Groups}
\author{F. G\'omez-Ortiz}
\email[Corresponding author: ]{fgomez@uliege.be}
\affiliation{Physique Th\'eorique des Mat\'eriaux, Q-MAT, Universit\'e de Li\`ege, B-4000 Sart-Tilman, Belgium} 
\author{S. Mamoudou Taganga}
\affiliation{Physique Th\'eorique des Mat\'eriaux, Q-MAT, Universit\'e de Li\`ege, B-4000 Sart-Tilman, Belgium}
\author{Emma E. McCabe}
\affiliation{Department of Physics, Durham University, South Road, Durham, DH1 3LE, U. K.}
\author{A. H. Romero}
\affiliation{Department of Physics and Astronomy, West Virginia University, Morgantown, WV 26505-6315, USA}
\author{E. Bousquet}
\email[Corresponding author: ]{eric.bousquet@uliege.be}
\affiliation{Physique Th\'eorique des Mat\'eriaux, Q-MAT, Universit\'e de Li\`ege, B-4000 Sart-Tilman, Belgium} 
\date{\today}
\begin{abstract}
With growing interest in structural chirality in periodic solids, it has been suggested that chirality might constitute a new ferroic order parameter. In this work, we demonstrate, by means of a formal group-theoretical proof and a systematic group–subgroup analysis, that achiral-to-chiral transitions that produce either member of an enantiomorphic pair cannot be driven by a Brillouin-zone-center ($\Gamma$-point) instability from a common achiral parent. We further substantiate this result by explicitly showing that none of the achiral parent space groups that admit symmetry-chiral phonon eigenvectors host them at the zone center. Given that a primary ferroic order parameter must, among other requirements, transform according to a zone-center irreducible representation, we conclude that phase transitions leading to enantiomorphic space groups cannot be classified as primary ferroic transitions. This predicts that any critical enhancement of fluctuations must occur at finite-$q$ rather than as a divergence of a uniform macroscopic susceptibility.
\end{abstract}
\maketitle

The study of chirality in crystalline solids has gained renewed momentum in recent years. This resurgence has been driven in part by the identification of chiral phonons~\cite{Zhu-18,Juraschek-25}, as well as by the growing recognition that chirality plays a significant role in several high-impact areas of condensed-matter physics, including topological insulators~\cite{Chang-18,Sanchez-19}, magnetic skyrmions~\cite{Muhlbauer-09}, and unconventional transport phenomena in noncentrosymmetric materials~\cite{Juraschek-19,Li-20}. Beyond being a geometric descriptor, chirality has emerged as a structural feature capable of shaping electronic, vibrational, and magnetic properties.

Recent reviews~\cite{Felser-22,Bousquet-25} have emphasized chirality as an organizing principle~\cite{Hlinka-14} for emergent behavior in solids. 
In particular, chirality has been discussed in the context of spontaneous symmetry breaking. It has been proposed that structural transitions from achiral to chiral phases may arise from the condensation of soft phonon modes~\cite{Bousquet-25,Fava-25, GomezOrtiz-25}. 
Within such a framework, chirality can be treated as an order parameter, and Landau free-energy expansions have been proposed to describe the energy landscape connecting left- and right-handed configurations~\cite{Bousquet-25,Spaldin-26}. 
Motivated by these developments, it has also been hypothesized that chirality itself might be treated in close analogy with conventional ferroic order parameters~\cite{Fava-25}. 
In particular, this question has attracted special attention in the case of the 22 enantiomorphic space groups.
With this perspective, the two enantiomorphic variants are viewed as symmetry-related states separated by a double-well potential~\cite{Fava-25, Bousquet-25}. 
This formulation naturally prompts the search for a \textit{universal} conjugate field that would couple linearly to chirality~\cite{Bousquet-25, fava2025prl}, allowing selective stabilization of a chosen enantiomorph and opening the highly desirable scenario of externally controllable “ferrochiral” behavior~\cite{Bousquet-25,Spaldin-26}.

Following the phenomenological definition of ferroicity~\cite{wadhawan2000introduction, toledano2015, Toledano-87}, for chirality to qualify as a primary ferroic order, it needs to satisfy four criteria~\cite{wadhawan2000introduction}: (i) the spontaneous emergence of long-range order below a critical temperature (T$_\text{C}$), (ii) the existence of enantiomorphic domains that exhibit hysteresis when switched, (iii) the ability to undergo field-induced phase transitions, and (iv) a characteristic enhancement of macroscopic susceptibility near (T$_\text{C}$).
In the framework of Landau theory, this further indicates that the primary ferroic transition has to take place at the center of the Brillouin zone ($k=0$), ensuring that the order parameter is macroscopic and preserves the unit cell size~\cite{toledano2015, Toledano-87, Toledano-77}.

This primary ferroicity has been questioned regarding the gyrotropic properties, e.g., for gyroelectric crystals, where switching the ferroelectric order switches the natural optical activity of the crystal.
Optical gyrotropy has then been viewed as an implicit form of ferroicity (often acting as a by-product of ferroelectricity or ferroelasticity)~\cite{wadhawan1979}; classifying it as an independent primary ferroic order would, however,  require the identification of an unique ordering field that can directly couple to the chiral order parameter~\cite{wadhawan2000introduction}.
A similar debate has recently arisen in the context of the magnetic toroidal moment~\cite{sannikov2007, vanaken2007, zimmermann2014}, where the question of primary ferroicity was carefully examined and ultimately resolved with the recognition of ferrotoroidicity as a genuine ferroic order~\cite{toledano2015, gnewuch2019}.
In this sense, the present discussion regarding chirality follows an established line of inquiry. 
The central symmetry question is: can a single achiral parent phase generate both enantiomorphs of an enantiomorphic pair through a $\Gamma$-point order parameter?

In this work, we address this question at the level of group theory and group-subgroup relationships. 
We demonstrate, in a general way, that achiral-to-chiral phase transitions among the 22 enantiomorphic space groups cannot arise from $\Gamma$-point irreducible representations, generalizing a result previously noted for chiral spinels~\cite{Talanov-15,Talanov-16}. 
Instead, such transitions necessarily originate from instabilities at non-zero wave vectors. 
We further corroborate this result by examining the symmetry properties of chiral phonon modes and showing that none of the achiral parent space groups hosting chiral phonons possess them at the zone center.
Our result implies that structural chirality leading to enantiomorphic space groups cannot be classified as a primary ferroic order. 
We also discuss in the conclusion the case of the 43 Sohnke groups that hold chiral crystals, but where, contrary to the enantiomorphic space groups, the left- and right handed structures belong to the same space group. 
These results specify the circumstances in which drawing an analogy to ferroic behavior is, or is not, justified.

\emph{---Zone-Center Exclusion Theorem in chiral transitions to Enantiomorphic pairs---}

\begin{theorem}\label{Theorem}
    Let $L_1$ and $L_2$ be an enantiomorphic pair. Then, any space group $H$ that is a common supergroup of both enantiomorphs must satisfy $[T_H:T]>1$.
\end{theorem}
\begin{proof}
Since $L_1$ and $L_2$ are enantiomorphic space groups, they share the same point group and translational subgroup $T$.

Without loss of generality ($L_1$ and $L_2$ space groups each include a screw axis $n_p$ of opposite helicity), we may assume that the screw operation
\[
n_p = (C_n \mid \tfrac{p}{n}c) \in L_1,
\]
for some integer $p < \tfrac{n}{2}$, while its enantiomorphic counterpart
\[
n_{n-p} = (C_n \mid \tfrac{n-p}{n}c) \in L_2.
\]

Since both $L_1$ and $L_2$ are subgroups of $H$, it follows that
\[
n_{n-p} n_p^{-1} \in H.
\]

A direct computation yields
\begin{align}
n_{n-p} n_p^{-1}
&= (C_n \mid \tfrac{n-p}{n}c)\,(C_n^{-1} \mid -\tfrac{p}{n}c) \nonumber \\
&= (E \mid \tfrac{n-2p}{n}c) \in T_H\trianglelefteq H.
\end{align}

Because $p < \tfrac{n}{2}$, the quantity $\tfrac{n-2p}{n}c$ corresponds to a non-lattice fractional translation of $L_1$ and $L_2$ that must lie in the translational subgroup $T_H$ of $H$, implying that
\[
[T_H : T] > 1.
\]
\end{proof}

It immediately follows from the theorem that transitions from any parent structure to chiral enantiomorphic pairs cannot occur at the $\Gamma$ point.
Indeed, a $\Gamma$-point instability preserves the translational subgroup of the supergroup $H$, whereas the theorem shows that for an enantiomorphic transition $[T_H:T]>1$. Geometrically, this condition means that the enantiomorphic phases have a larger primitive unit cell than of the parent (of space group $H$). 
Consequently, transitions from achiral to chiral enantiomorphic pairs must involve finite-wavevector instabilities.
We also clearly see that the instability is always commensurate. 
The translational index obtained from the group–subgroup relation fixes the required multiplicity of the supercell according to the theorem 
$(n-2p)/n$. 
For example, transitions towards $P6_2~\&~P6_4$ will involve a $k$-index of 3, [$(6-2\cdot 2)/6=1/3$], whereas transitions towards $P4_1~\&~P4_3$ will involve a $k$-index of 2, [$(4-2\cdot 1)/4=1/2$] as illustrated in Figure.~\ref{fig:sketch}. 
\begin{figure}[h]
     \centering
      \includegraphics[width=\columnwidth]{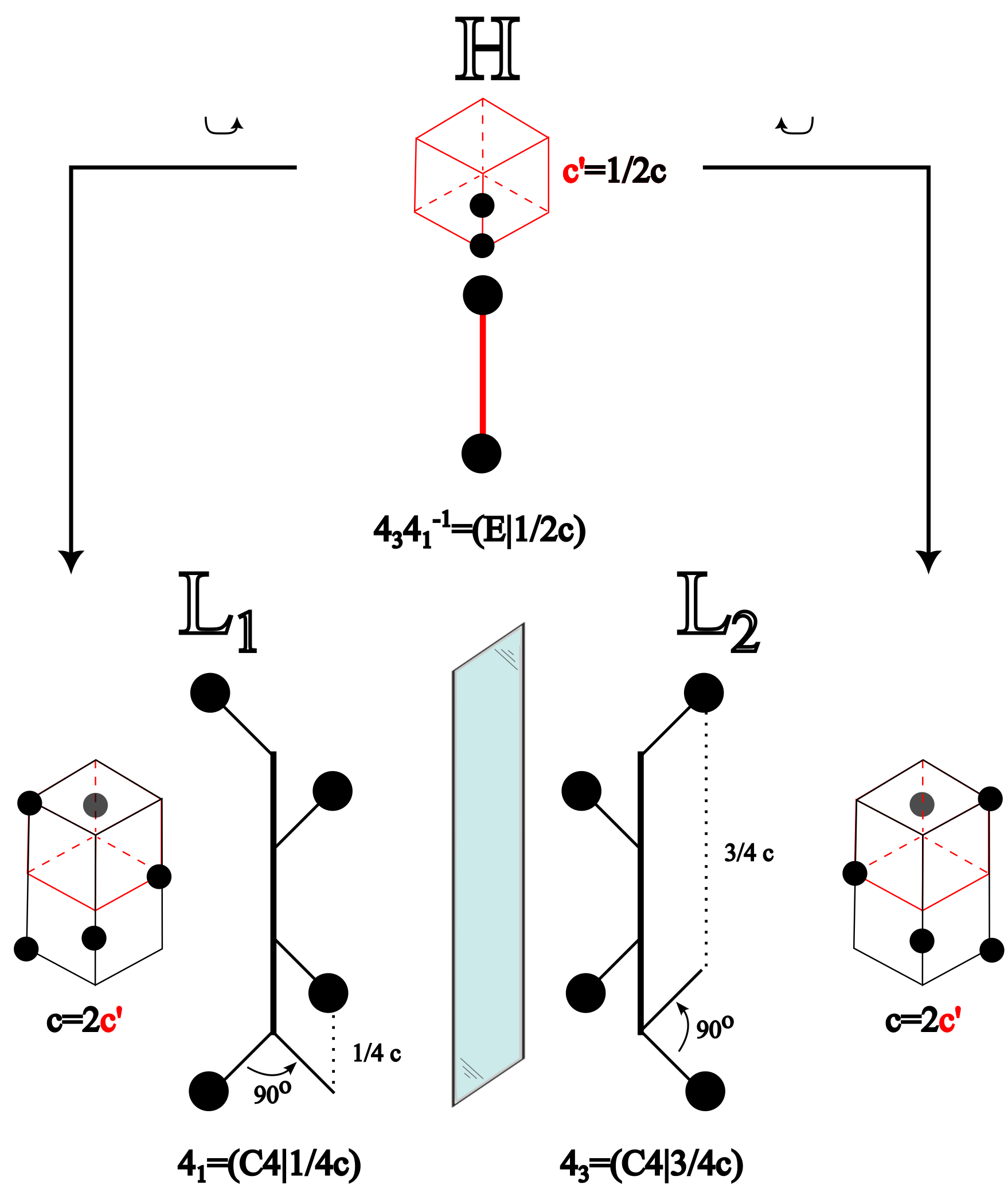}
      \caption{Illustrative example of the theorem applied to the enantiomorphic screw axis $P4_1$ and $P4_3$.} 
      \label{fig:sketch} 
\end{figure}
The participating soft modes must, therefore, condense at a commensurate finite wave vector. The specific wave vector depends on the space-group pair and is detailed in the Supplemental Material~\cite{supplemental_material}. For example, for the  P6$_2$/P6$_4$ and P4$_1$/P4$_3$ cases the corresponding commensurability is characterized by indices $3$ and $2$ respectively.

We emphasize that this result does not forbid zone-center instabilities in chiral systems more generally. In particular, $\Gamma$-point transitions can occur between chiral space groups. 
For example, in SiO$_2$, a zone-center instability transforms the system from the chiral enantiomorphic space group $P6_422$ to another chiral enantiomorphic space group $P3_121$~\cite{Shapiro-67,Grimm-75,Gomez-Ortiz26}. However, such a $\Gamma$-point transition selects a single handedness and cannot produce the opposite enantiomer ($P3_221$, which cannot be reached through the same order parameter) nor through an achiral parent phase. As a result, such a transition does not constitute a ferroic switching of chirality and lies outside the class of processes considered in our theorem. Indeed, we can see that $6_4=(3_1)^2$ and therefore both space groups share the same translational subgroup.
Similarly, transitions between achiral and non-enantiomorphic chiral space groups can occur at the zone center, as in Pb$_5$Ge$_3$O$_{11}$~\cite{Fava-24}, Ba(TiO)Cu$_4$(PO$_4$)$_4$~\cite{Hayashida-21}, or AgNbO$_3$~\cite{Song-26,Safari-26}, where $\Gamma$-point instabilities drive a transition from $P\bar{6}$ to $P3$; $P4/nmm$ to $P42_12$, or $Im\bar{3}$ to $R3$ respectively. 
Together, these examples demonstrate that only transitions from a parent phase to both members of an enantiomorphic pair are forbidden at the zone center, in agreement with our theorem; by contrast, transitions to non-enantiomorphic (``potato-like''~\cite{Ruch-77,Bruce-03,Felser-22}) chiral space groups or those selecting a single handedness remain allowed.
In these later transitions to nonenantiomorphic Sohnke groups, the structure is not necessarily chiral; rather the space group does not contain improper symmetry operations (and so a chiral object cannot be wrapped onto its enantiomorph).
This is consistent with the possibility of a transition involving a zone-center mode which simply breaks symmetry elements.

\emph{--- Confirmation via Space-Group and Chiral-Phonon Analysis ---}
To provide empirical support for the theorem, we analyzed the space-group relations between parent supergroups and their enantiomorphic pair using the utilities available in the Bilbao Crystallographic Server~\cite{Aroyo-06,Aroyo-06.2} and the ISOTROPY software suit~\cite{Hatch-03}. 
A systematic search for common supergroups of each enantiomorphic pair confirms that all common supergroups of the enantiomorphic pair have a $k$-index larger than 1, meaning that at least a cell doubling is at play. The complete enumeration is provided in Tab.~\ref{tab:enantiomorphic_kindex}. 
\begin{table}[hbp]
\centering
\begin{tabular}{cc}
\hline
\textbf{Enantiomorphic pair} & \textbf{$k$-index} \\
\hline
P$3_1$/P$3_2$ & 3 \\
P$4_1$/P$4_3$ & 2 \\
P$6_1$/P$6_5$ & 3 \\
P$6_2$/P$6_4$ & 3 \\
P$4_122$/P$4_322$ & 2 \\
P$4_12_12$/P$4_32_12$ & 2 \\
P$6_122$/P$6_522$ & 3 \\
P$6_222$/P$6_422$ & 3 \\
P$3_112$/P$3_212$ & 3 \\
P$3_121$/P$3_221$ & 3 \\
P$4_132$/P$4_332$ & 2 \\
\hline
\end{tabular}
\caption{Enantiomorphic space-group pairs and the minimum $k$-index of any common parent supergroup.}
\label{tab:enantiomorphic_kindex}
\end{table}

Additionally, We examined all achiral space groups capable of hosting a chiral phonon whose irreducible representation can yield an enantiomorphic space group. 
These possibilities are listed in Tables SI to SVII of the Supplementary Information~\cite{supplemental_material}. 
Our exhaustive analysis reveals that, in all symmetry-allowed cases, chiral phonon modes occur exclusively away from the $\Gamma$ point. 
This confirms that zone-center instabilities for enantiomorphic space groups do not exist, thereby providing explicit validation of Theorem 1.

We further observe that trigonal non-enantiomorphic space groups always require a $k$-index $\geq 3$, whereas tetragonal groups require a $k$-index $\geq 2$. 
Finally, our analysis shows that hexagonal and rhombohedral achiral space groups never exhibit a chiral irreducible representation at a strict zone boundary point—unlike tetragonal and cubic cases, where such zone boundary points are present.
Instead, chiral to enantiomorphic phase transitions in these systems are driven by non-zero wave vectors between $\Gamma$ and zone boundary points. An example of such a transition is provided by the CsCuCl$_3$ compound~\cite{Schlueter1966,Hirotsu1975}, where a $\delta_6$ mode with $q=(0,0,\frac{1}{3})$ drives the system from the $P6_3/mmc$ parent phase to the enantiomorphic pair $P6_122/P6_522$.

\emph{--- Physical implications ---}
The exclusion of $\Gamma$-point instabilities for transitions towards enantiomorphic-pairs directly impacts the proposed analogy between chirality and conventional ferroic order parameters.
In proper ferroics, the order parameter transforms according to a zone-center irreducible representation, which permits linear coupling to a homogeneous external (conjugate) field capable of inducing its reversal.
Therefore, an immediate implication is that no universal homogeneous conjugate field can couple linearly to the primary chiral order parameter for enantiomorphic-pair transitions, since the primary instability is necessarily finite-$q$.
A central consequence of our result is thus the absence of a universal conjugate field for chirality switching, which is instead intrinsically material-specific and must be treated on a case by case basis.
This perspective is further supported by the case of non-enantiomorphic space groups, where chirality can arise from a $\Gamma$-point distortion. 
Even in such situations, chirality typically emerges from the coupling between an axial and a polar mode~\cite{Hayashida-21,Fava-24} suggesting that chirality is inherently second-order and its reversal depends on the specific microscopic modes involved. 
A related implication concerns the nature of the critical behavior. 
Since the soft mode condenses at a finite wave vector, critical fluctuations occur at nonzero $k$. 
Consequently, one should not expect a divergence of any uniform ($q=0$) susceptibility associated with chirality. Instead, the relevant critical response is the generalized susceptibility at the ordering wave vector $q*$~\cite{Scott-12}, and can be observed via supercell peaks or diffuse scattering for that particular $q$-vector.

Beyond these general considerations, the structure of the possible irreducible representations collected in Supplementary Table SI to SXVII~\cite{supplemental_material} provides additional insight into the nature of the order parameter.
The first interesting observation is the multidimensional character of the participating irreducible representations. The resulting multiplicity of symmetry-allowed domain directions implies that opposite handedness cannot be described as a simple sign reversal of a scalar amplitude (echoing Kelvin's description of chirality in terms of non-sperimposable mirror images~\cite{Bousquet-25}).
Rather, the enantiomorphic states correspond to distinct orientations within a higher-dimensional order-parameter space.
As a consequence, chiral domain walls are not generically ``Ising-like'' objects associated with the inversion of a one-component order parameter. 
Instead, they involve reorientation within a multidimensional manifold, potentially allowing for continuous textures or phase rotations across the wall. 
This symmetry structure, beyond conventional Ising-type interfaces, provides a natural setting for experimental investigations using real-space imaging and scattering probes, and suggests the possibility of observing topological defects beyond conventional domain walls classified by $\pi_0(\mathcal{X})$, including defects associated with nontrivial higher homotopy groups of the order-parameter space, such as vortices.

The second, even more surprising situation arises when the common parent phase belongs to a Sohncke space group that is non-enantiomorphic (i.e., it lacks improper rotation symmetry but is not inherently chiral). 
For those Sohncke groups, we have identified explicit examples in which opposite handedness can be linked to the condensation of two distinct zone boundary irreducible representations (e.g. $P4_2$ space group where $Z_3$ and $Z_4$ give the $P4_1$ and $P4_3$ enantiomorphic space groups, respectively~\cite{GomezOrtiz-25}).
%
Since the two resulting enantiomorphic structures are mirror images of each other, their common parent structure (recovered when the order parameter amplitude vanishes) must itself be invariant under mirror symmetry. 
This means that the parent phase, despite belonging to a chiral Sohncke space group, must contain an achiral structural motif.

Contrary to the enantiomorphic space groups, where enantiomorphs arise from different orientations of the same irreducible representation (i.e. different order parameter directions), here they originate from entirely separate irreducible representation.
To preserve chiral symmetry in the parent achiral phase, these channels must be exactly degenerate, e.g. as reported in Ref.~\cite{GomezOrtiz-25} for Sr$_2$As$_2$O$_7$ in the $P4_2$ space group, where the $Z_3$ and $Z_4$ unstable phonon modes have the same frequency and give the same gain of energy.
This degeneracy between distinct irreducible representations typically suggests a first-order phase transition as the system must choose between competing  symmetry-distinct instabilities~\cite{Toledano-87}.
Finally, domain walls between oppositely handed regions in this scenario are unique as they will involve the suppression of one instability and the concurrent growth of the other. 
This behavior is qualitatively different from conventional ferroic domain walls and highlights the complex character of structural chirality.

\emph{---Conclusions---}
We have demonstrated that structural transitions from achiral phases to enantiomorphic space-group pairs must exclude Brillouin-zone center instabilities. 
This symmetry-imposed constraint demonstrates that the recently invoked analogy between chirality and primary ferroicity\cite{Bousquet-25,Spaldin-26} cannot hold in the case of enantiomorphic transitions and that intrinsic differences separate the two phenomena.
While the comparison is suggestive (presence of a double well between two opposite states, i.e. left and right handed structures that can be seen as ``up'' and ``down'' states of a ferroic order, and connected by a high symmetry reference structure), at a phenomenological level this analogy does not resist  formal symmetry analysis for enantiomorphic space-group pairs. 
The emergence of handedness in these transitions is therefore intrinsically distinct from the mechanism underlying primary ferroic phase transitions as classically defined ~\cite{wadhawan2000introduction, toledano2015}.

The situation is, however, different for the 43 non-enantiomorphic Sohncke space groups in which zone-center transitions have been reported (see Supporting Information~\cite{supplemental_material}). Observations of chiral phase transitions toward these 43 space groups suggest that chirality typically emerges as a byproduct, arising from the coupling between, for example, axial (respectively antiaxial) and polar (respectively antipolar) orders~\cite{fabini2024,Hayashida-21,Fava-24,Song-26,Safari-26,wadhawan1979,Aramberri-25}. 
This issue has been investigated in the context of gyroelectrics (or ferrogyrotropy), where optical activity was found to be a consequence of other primary order parameters. 
For chirality itself, however, the problem still requires a more detailed investigation~\cite{Bousquet-25}. 
To the best of our knowledge, there is only one reported case in which chirality arises as a primary order parameter in bulk material~\cite{Hayashida-21}, thus constituting the only known example of a proper chiral transition. 
In that case, the mode simultaneously involves both antiaxial and antipolar motions inside the same eigenvector. 
Nevertheless, the unit cell is relatively large, containing several motif units allowing for coupled complex internal distortion patterns even within a zone-center mode. 
Even in the case of non-enantiomorphic space groups, where zone-center distortions are allowed, no unambiguous conjugate field capable of selectively coupling to and switching structural chirality has, to date, been demonstrated. This remains an open question to be scrutinized, despite recent reports of ultrafast light-induced control of chirality~\cite{Romao2024}. 
\acknowledgments
We acknowledge He Xu for his constructive feedback and stimulating discussions that enriched our analysis.
F.G.O. acknowledges financial support from MSCA-PF 101148906 funded by the European Union and the Fonds de la Recherche Scientifique (FNRS) through the FNRS-CR 1.B.227.25F grant. S.M.T.  acknowledges the Fonds de la Recherche Scientifique (FNRS) for financial support with a FRIA doctoral grant.
F.G.-O., S.M.T and E.B. acknowledge the Fonds de la Recherche Scientifique (FNRS) for financial support, the PDR project CHRYSALID No.40003544, the Consortium des \'Equipements de Calcul Intensif (C\'ECI), funded by the Fonds de la Recherche Scientifique (F.R.S.-FNRS) under Grant No. 2.5020.11 and the Tier-1 Lucia supercomputer of the Walloon Region, infrastructure funded by the Walloon Region under the grant agreement n$^\circ$1910247.
\newpage
\onecolumngrid
\section{Supplementary Information for Chirality Cannot Be Ferroic in Enantiomorphic Space-Groups}
\subsection{Chiral phonons from achiral space-groups}
In this section, we systematically enumerate all symmetry-allowed pathways that connect achiral space groups to enantiomorphic-pair families. For each case, we identify the corresponding $q$-point at which the chiral phonon emerges and specify the domain direction associated with the condensation into each enantiomorphic space group.
In some cases, the large number of possible domain directions associated with the irreducible representation is omitted from the table to maintain compactness.
\begin{table}
    \begin{tabular}{c|ccc}
    \hline\hline
     high sym. achiral sg.    & IRREP & domain direction & enantiomorphic sg. \\
     \hline 
     $P4_2/m$ (84)   & $Z_2$ & $(a,0)$ & $P4_1$ (76)  \\
     $P4_2/n$ (86)   & $Z_2$ & $(0,a)$ & $P4_3$ (78)  \\
\hline
\multirow{4}{*}{$I4_1$ (88)} & $X_1$ & (0,a,0,a); (0,a,0,-a) & \multirow{2}{*}{$P4_1$ (76)} \\
            & $M_2$ & (a,0) & 
             \\ \cline{2-4}
           & $X_1$ & (a,0,a,0); (a,0,-a,0) & \multirow{2}{*}{$P4_3$ (78)} \\
            & $M_2$ & (0,a) & \\   
            \hline
    \multirow{2}{*}{$P4_2nm$ (102)} & $Z_5$ & $(a,0)$ & $P4_3$ (78) \\
    & $Z_5$ & $(0,a)$ & $P4_1$ (76) \\
    \hline
    $P4_2mc$ (105)   & $Z_5$ & $(a,0)$ & $P4_1$ (76)  \\
     $P4_2bc$ (106)  & $Z_5$ & $(0,-a)$ & $P4_3$ (78)  \\
     \hline
\multirow{4}{*}{$I4_1md$ (109)} & $X_1$ & (a,0,0,a); (a,0,0,-a) & \multirow{2}{*}{$P4_1$ (76)} \\
            & $M_5$ & (a,0) & 
             \\ \cline{2-4}
           & $X_1$ & (0,a,a,0); (0,-a,a,0) & \multirow{2}{*}{$P4_3$ (78)} \\
            & $M_5$ & (0,a) & \\   
\hline
\multirow{4}{*}{$I4_1cd$ (110)} & $X_1$ & (a,0,0,a); (a,0,0,-a) & \multirow{2}{*}{$P4_3$ (78)} \\
            & $M_5$ & (a,0) & 
             \\ \cline{2-4}
           & $X_1$ & (0,a,a,0); (0,-a,a,0) & \multirow{2}{*}{$P4_1$ (76)} \\
            & $M_5$ & (0,a) & \\   
\hline
                        & $Z_3$ & $(a,0)$ & \multirow{2}{*}{$P4_122$ (91)}  \\
     $P4_2/mmc$ (131)   & $Z_4$ & $(a,0)$ &   \\

     $P4_2/nbc$ (133)   & $Z_3$ & $(0,-a)$ & \multirow{2}{*}{$P4_322$ (95)}  \\
                        & $Z_4$ & $(0,a)$  &   \\
\hline
                        & $Z_3$  & $(a,a)$  & \multirow{2}{*}{$P4_322$ (95)}  \\
     $P4_2/mcm$ (132)   & $Z_4$  & $(-a,a)$ &   \\
     $P4_2/nnm$ (134)   & $Z_3$  & $(-a,a)$  & \multirow{2}{*}{$P4_122$ (91)}    \\
                        & $Z_4$  & $(a,a)$  &    \\
                        \hline
                        & $Z_3$  & $(a,0)$  & \multirow{2}{*}{$P4_12_12$ (92)}  \\
     $P4_2/mbc$ (135)   & $Z_4$  & $(a,0)$ &   \\
     $P4_2/nmc$ (137)   & $Z_3$  & $(0,-a)$  & \multirow{2}{*}{$P4_32_12$ (96)}    \\
                        & $Z_4$  & $(0,a)$  &    \\
\hline
           & $X_1$ & (a,0,a,0); (-a,0,a,0) & \multirow{2}{*}{$P4_122$ (91)} \\
            & $M_3$ & (a,0) & 
             \\ \cline{2-4}
                    & $X_1$ & (0,a,0,a); (0,-a,0,a) & \multirow{2}{*}{$P4_322$ (95)} \\
$I4_1/amd$ (141)    & $M_3$ & (0,-a) & \\  \cline{2-4}
$I4_1/acd$ (142)    & $X_2$ & (a,0,a,0); (-a,0,a,0) & \multirow{2}{*}{$P4_12_12$ (92)} \\
            & $M_4$ & (a,0) & 
             \\ \cline{2-4}
           & $X_2$ & (0,a,0,a); (0,-a,0,a) & \multirow{2}{*}{$P4_32_12$ (96)} \\
            & $M_4$ & (0,a) & \\  
                                        
\hline
    \end{tabular}
    \caption{Achiral tetragonal space groups (first column) that hold chiral phonons with a zone boundary irreducible representation (second column) capable of generating an enantiomorphic subgroup (fourth column). The third column lists the domain direction that is associated with the enantiomorphic subgroup of column four.} 
\end{table}

\begin{table*}
    \begin{tabular}{c|ccc}
    \hline
     high sym. achiral sg.   & IRREP & domain direction & enantiomorphic sg. \\
     \hline 

\hline
    \multirow{3}{*}{$Pm\bar{3}n$ (223)} & \multirow{3}{*}{$X_3$, $X_4$} & (a,0;0,0;0,0) & $P4_122$ (91)/$P4_322$ (95) \\
                                        &  & (a,b;a,-b;c,0) & $P4_12_12$ (92)/$P4_32_12$ (96) \\              
                                        &  & (a,0;a,0;a,0) & $P4_332$ (212)/$P4_132$ (213)  \\

\hline
    \multirow{13}{*}{$Fd\bar{3}m$ (227)} & $X_1$ & (a,a;a,-a;0,0) & $P4_122$ (91)/$P4_322$ (95) \\\cline{2-4}
                                        & $X_2$ & (a,a;a,-a;0,0) & $P4_122$ (91)/$P4_322$ (95) \\\cline{2-4}
                                        & \multirow{3}{*}{$X_3$} & (a,0;0,0;0,0) & $P4_122$ (91)/$P4_322$ (95) \\
                                        &       & (a,0;a,0;b,0) & $P4_12_12$ (92)/$P4_32_12$ (96) \\
                                        &       & (a,0;a,0;a,0) & $P4_332$ (212)/$P4_132$ (213) \\\cline{2-4}
                                        & \multirow{2}{*}{$X_4$} & (a,0;a,0;0,0) & $P4_12_12$ (92)/$P4_32_12$ (96) \\
                                        &       & (a,0;0,b;a,0) & $P4_1$ (76)/$P4_3$ (78) \\ \cline{2-4}
                                        & \multirow{3}{*}{$W_1$} & (a,0,0,0;0,a,0,0;0,0,0,0) & $P4_122$ (91)/$P4_322$ (95)  \\
                                        &       & (a,0,a,0;0,a,0,a;0,0,0,0) & $P4_12_12$ (92)/$P4_32_12$ (96)  \\
                                        &       & (a,b,a,b;0,a,$\sqrt{2}$b,a;0,0,0,0) & $P4_1$ (76)/$P4_3$ (78) \\ \cline{2-4}
                                        & \multirow{3}{*}{$W_2$} & (a,0,0,0;0,a,0,0;0,0,0,0) & $P4_122$ (91)/$P4_322$ (95)  \\
                                        &       & (a,0,a,0;0,a,0,a;0,0,0,0) & $P4_12_12$ (92)/$P4_32_12$ (96)  \\
                                        &       & (a,b,a,b;0,a,$\sqrt{2}$b,a;0,0,0,0) & $P4_1$ (76)/$P4_3$ (78) \\ 
\hline
    \multirow{10}{*}{$Fd\bar{3}c$ (228)} & $X_1$ & (a,a;a,-a;0,0) & $P4_122$ (91)/$P4_322$ (95) \\\cline{2-4}
                                        & $X_2$ &  (a,a;a,-a;0,0) & $P4_122$ (91)/$P4_322$ (95) \\\cline{2-4}
                                        & \multirow{3}{*}{$X_3$} & (a,0;0,0;0,0) & $P4_122$ (91)/$P4_322$ (95) \\
                                        &       & (a,0;a,0;b,0) & $P4_12_12$ (92)/$P4_32_12$ (96) \\
                                        &       & (a,0;a,0;a,0) & $P4_332$ (212)/$P4_132$ (213) \\\cline{2-4}  
                                        & \multirow{2}{*}{$X_4$} & (a,0;a,0;0,0)/(a,0;0,0;0,0) & $P4_12_12$ (92)/$P4_32_12$ (96) \\
                                        &       & (a,0;0,b;a,0) & $P4_1$ (76)/$P4_3$ (78) \\ \cline{2-4} 
                                        & \multirow{3}{*}{$W_1W_2$} & & $P4_122$ (91)/$P4_322$ (95)  \\
                                        &       & & $P4_12_12$ (92)/$P4_32_12$ (96)  \\
                                        &       & & $P4_1$ (76)/$P4_3$ (78) \\ 
                                        \hline
    \multirow{5}{*}{$Ia\bar{3}d$ (230)} & $H_1$ & (0,a)/(a,0) & $P4_332$ (212)/$P4_132$ (213) \\
                                        & $H_4$ & (a,a,0,0,0,0) & $P4_122$ (91)/$P4_322$ (95) \\
                                        & $H_2H_3$ & (a,b,-b,a) & $P4_12_12$ (92)/$P4_32_12$ (96)  \\
                                        & $N_1$ & (a,0;a,0;0,0;0,0;0,0;0,0) & $P4_12_12$ (92)/$P4_32_12$ (96)  \\
                                        & $N_2$ & (a,0;a,0;0,0;0,0;0,0;0,0) & $P4_122$ (91)/$P4_322$ (95) \\
                                        
\hline
    \end{tabular}
    \caption{Table II continued for achiral cubic space groups. In view of the large number of possible domains for the achiral cubic space groups cases, only one or none is shown for clarity.}
\end{table*}


\begin{table*}
    \begin{tabular}{c|ccc}
    \hline
     high sym. achiral sg.   & IRREP & domain direction & enantiomorphic sg. \\
     \hline 

\hline
    \multirow{3}{*}{$P\bar{3}$ (147)} & $DT2DT3$& $(0,b,0,a)$ & $P3_1$ (144) \\
                                     & $DT2DT3$ &  $(a,0,b,0)$  & $P3_2$ (145) \\
                                     & $U1$&  &  $P3_1$ (144)/$P3_2$ (145)\\
\hline
    \multirow{4}{*}{$R\bar{3}$ (148)} & $LD2LD3$& $(0,b,0,a)$ & $P3_1$ (144) \\
                                        & $LD2LD3$ &  $(a,0,b,0)$  & $P3_2$ (145) \\
                                        & $SM1$&  & \multirow{2}{*}{$P3_1$ (144)/$P3_2$ (145)} \\
                                        & $Y1$&  &  \\
\hline
    \multirow{4}{*}{$P3m1$ (156)}   &  $DT3DU3$ & $(a,0,b,0)$ & $P3_1$ (144) \\
                                    & $DT3DU3$ &  $(0,-\cos{(\frac{\pi}{3})}a+\sin{(\frac{\pi}{3})}b,0,-\sin{(\frac{\pi}{3})}a-\cos{(\frac{\pi}{3})}b)$ & $P3_2$ (145) \\
                                    & $U1UA1$&  & \multirow{2}{*}{$P3_1$ (144)/$P3_2$ (145)} \\
                                    & $U2UA2$&  & \\
\hline
    \multirow{5}{*}{$P31m$ (157)}   & $DT3DU3$& $(a,0,b,0)$ & $P3_1$ (144)\\
                                    & $DT3DU3$ &  $(0,-\cos{(\frac{\pi}{3})}a+\sin{(\frac{\pi}{3})}b,0,-\sin{(\frac{\pi}{3})}a-\cos{(\frac{\pi}{3})}b)$ & $P3_2$ (145) \\
                                    & $U1UA1$&  & \multirow{3}{*}{$P3_1$ (144)/$P3_2$ (145)}\\
                                    & $U2UA2$&  & \\
                                    &$D1DA1$&   & \\
\hline
    \multirow{4}{*}{$P3c1$ (158)}   & $DT3DU3$ & $(a,0,b,0)$ & $P3_1$ (144) \\
                                    & $DT3DU3$ &  $(0,-a,0,-b)$  & $P3_2$ (145) \\
                                    & $U1UA1$&  & \multirow{2}{*}{$P3_1$ (144)/$P3_2$ (145)} \\
                                    & $U2UA2$&  & \\
\hline
    \multirow{5}{*}{$P31c$ (159)}   & $DT3DU3$& $(a,0,b,0)$ & $P3_1$ (144)\\
                                    & $DT3DU3$ &  $(0,-a,0,-b)$  & $P3_2$ (145) \\
                                    & $U1UA1$&  & \multirow{3}{*}{$P3_1$ (144)/$P3_2$ (145)}\\
                                    & $U2UA2$&  & \\
                                    &$D1DA1$&   & \\
\hline
    \multirow{5}{*}{$R3m$ (160)}    & $LD3LE3$& $(a,0,b,0)$ & $P3_1$ (144)\\
                                    & $LD3LE3$&  $(0,-\cos{(\frac{\pi}{3})}a-\sin{(\frac{\pi}{3})}b,0,\sin{(\frac{\pi}{3})}a-\cos{(\frac{\pi}{3})}b)$  & $P3_2$ (145) \\
                                    & $SM1$&  & \multirow{3}{*}{$P3_1$ (144)/$P3_2$ (145)} \\
                                    & $Y1$&  & \\
                                    & $GPGQ1$& & \\
\hline
    \multirow{5}{*}{$R3c$ (161)}    & $LD3LE3$ & $(a,0,b,0)$ & $P3_1$ (144)\\
                                    & $LD3LE3$ &  $(0,\cos{(\frac{\pi}{3})}a-\sin{(\frac{\pi}{3})}b,0,\sin{(\frac{\pi}{3})}a+\cos{(\frac{\pi}{3})}b)$  & $P3_2$ (145) \\
                                    & $SM1$&  & \multirow{3}{*}{$P3_1$ (144)/$P3_2$ (145)} \\
                                    & $Y1Y1$&  & \\
                                    & $GPGQ1$& & \\
\hline
                                    & $DT3$ & $(0,-a,0,b)$ & $P3_1$ (144)\\
                                    & $DT3$ & $(a,0,b,0)$ & $P3_2$ (145)\\
\multirow{2}{*}{$P\bar{3}1m$ (162)} & $U1,U2$&  & \multirow{2}{*}{$P3_1$ (144)/$P3_2$ (145)}\\
\multirow{2}{*}{$P\bar{3}1c$ (163)} & $D1$ &  & \\ \cline{2-4}
                                    & $DT3$ &  & \multirow{3}{*}{$P3_112$ (151)/$P3_212$ (153)} \\
                                    & $U1,U2$&  &  \\
                                    & $D1$&  &  \\              
\hline
    \end{tabular}
    \caption{Achiral trigonal space groups (first column) that hold chiral phonons from an irreducible representation that is not from the zone boundary (second column) and capable of generating an enantiomorphic subgroup (fourth column). The third column lists the domain direction that is associated with the enantiomorphic subgroup of column four, they are not all listed for clarity but they all have domain directions. }
\end{table*}

\begin{table*}
    \begin{tabular}{c|ccc}
    \hline
     high sym. achiral sg.   & IRREP & domain direction & enantiomorphic sg. \\
     \hline 
\hline
                                    & $DT3$ & $(0,-a,0,b)$ & \multirow{2}{*}{$P3_1$ (144)}\\
                                    & $P2P3$& $(0,b,0,a;d,0,c,0)$ & \\
                                    & $DT3$ & $(a,0,b,0)$ & \multirow{2}{*}{$P3_2$ (145)} \\
                                    & $P2P3$& $(a,0,b,0;0,c,0,d)$ & \\\cline{2-4}
    $P\bar{3}m1$ (164)              &$U1$ &  & \multirow{2}{*}{$P3_1$ (144)/$P3_2$ (145)}\\
    $P\bar{3}c1$ (165)              &$U2$ &  & \\\cline{2-4}
                                    & $P2P3$&  & $P3_112$ (151)/$P3_212$ (153) \\\cline{2-4}
                                    & $DT3$ &  & \multirow{3}{*}{$P3_121$ (152)/$P3_221$ (154)} \\
                                    &$U1$&  &  \\
                                    &$U2$&  &  \\
\hline
    \multirow{9}{*}{$R\bar{3}m$ (166)} 
                                    & $LD3$ & $(0,-a,0,b)$ & $P3_1$ (144)\\
                                    & $LD3$ & $(a,0,b,0)$ & $P3_2$ (145)\\\cline{2-4}
                                    & $SM2$ &  & \multirow{2}{*}{$P3_1$ (144)/$P3_2$ (145)}\\
                                    & $GP1$&  & \\ \cline{2-4}
                                    & $LD3$&  & $P3_121$ (152)/$P3_221$ (154) \\\cline{2-4}
                                    & $SM1$&  & \multirow{4}{*}{$P3_112$ (151)/$P3_212$ (153)} \\
                                    & $Y1$&  &  \\
                                    & $Y2$&  &  \\
                                    & $GP1$&  &  \\  
\hline
    \multirow{9}{*}{$R\bar{3}c$ (167)}
                                    & $LD3$ & $(0,-a,0,b)$ & $P3_1$ (144)\\
                                    & $LD3$ & $(a,0,b,0)$ & $P3_2$ (145)\\\cline{2-4}
                                    & $SM2$ &  & \multirow{3}{*}{$P3_1$ (144)/$P3_2$ (145)}\\
                                    & $Y1Y2$ & \\
                                    & $GP1$&  & \\\cline{2-4}
                                    & $LD3$&  & $P3_121$ (152)/$P3_221$ (154) \\\cline{2-4}
                                    & $SM1$&  & \multirow{3}{*}{$P3_112$ (151)/$P3_212$ (153)} \\
                                    & $Y1Y2$&  &  \\
                                    & $GP1$&  &  \\
\hline
    \multirow{4}{*}{$P\bar{6}$ (174)} & $DT2DT3$& $(0,b,0,a)$ & $P3_1$ (144)\\
                                      & $DT2DT3$& $(a,0,b,0)$ & $P3_2$ (145)\\\cline{2-4}
                                      & $U1$&  &\multirow{2}{*}{$P3_1$ (144)/$P3_2$ (145)}  \\
                                      & $D1DA1$&  &  \\
\hline
    \multirow{11}{*}{$P6/m$ (175)} & $DT4DT6$ & $(0,b,0,a)$ & \multirow{2}{*}{$P3_1$ (144)} \\
                                  & $P2P3$& $(0,b,0,a;0,d,0,c)$ & \\
                                  & $DT4DT6$ & $(a,0,b,0)$ & \multirow{2}{*}{$P3_2$ (145)}\\
                                  & $P2P3$& $(a,0,b,0;c,0,d,0)$ & \\\cline{2-4}
                                  & $U2$& & \multirow{2}{*}{$P3_1$ (144)/$P3_2$ (145)}\\
                                  & $D1$&  & \\\cline{2-4}
                                  & $DT3DT5$& $(a,0,b,0)$ & $P6_2$ (171) \\
                                  & $DT3DT5$& $(0,b,0,a)$ & $P6_4$ (172) \\\cline{2-4}
                                  & $P2P3$ &  & \multirow{3}{*}{$P6_2$ (171)/$P6_4$ (172)} \\
                                  & $U1$& &  \\
                                  & $D1$&  &  \\
\hline
    \end{tabular}
    \caption{Table IV continued. }
\end{table*}

\begin{table*}
    \begin{tabular}{c|ccc}
    \hline
     high sym. achiral sg.   & IRREP & domain direction & enantiomorphic sg. \\
     \hline 
\hline
    \multirow{11}{*}{$P6_3/m$ (176)} & $DT3DT5$& $(0,b,0,a)$ & \multirow{2}{*}{$P3_1$ (144)}\\
                                    & $P2P3$& $(0,b,0,a;0,d,0,c)$ & \\
                                    & $DT3DT5$& $(a,0,b,0)$ & \multirow{2}{*}{$P3_2$ (145)} \\
                                    & $P2P3$& $(a,0,b,0;c,0,d,0)$ & \\\cline{2-4}
                                    & $U1$ &  & \multirow{2}{*}{$P3_1$ (144)/$P3_2$ (145)} \\
                                    & $D1$ &  &  \\\cline{2-4}
                                    & $DT4DT6$& $(0,b,0,a)$ & $P6_1$ (169) \\
                                    & $DT4DT6$& $(a,0,b,0)$ & $P6_5$ (170) \\\cline{2-4}
                                    & $P2P3$&  & \multirow{3}{*}{$P6_1$ (169)/$P6_5$ (170)} \\
                                    & $U2$& & \\
                                    & $D1$ &  & \\
\hline
    \multirow{15}{*}{$P6mm$ (183)} & $DT5DU5$& $(0,-\cos{(\frac{\pi}{3})}a+\sin{(\frac{\pi}{3})}b,0,-\sin{(\frac{\pi}{3})}a-\cos{(\frac{\pi}{3})}b $& $P6_2$ (171)\\
                                   & $DT5DU5$& $(a,0,b,0)$& $P6_4$ (172)\\\cline{2-4}
                                  & $U1UA1$& & \multirow{5}{*}{$P6_2$ (171)/$P6_4$ (172)}\\\
                                  & $U2UA2$& & \\
                                  & $P3PA3$&  &  \\
                                  & $D1DA1$&  &  \\
                                  & $D2DA2$&  &  \\\cline{2-4}
                                  & $DT6DU6$& $(a,0,b,0)$ & \multirow{2}{*}{$P3_1$ (144)}\\
                                  & $P3PA3$& $(a,0,b,0;c,0,d,0)$ & \\
                                  & $DT6DU6$& $(0,-\cos{(\frac{\pi}{3})}a+\sin{(\frac{\pi}{3})}b,0,-\sin{(\frac{\pi}{3})}a-\cos{(\frac{\pi}{3})}b $ & \multirow{2}{*}{$P3_2$ (145)}\\
                                  & $P3PA3$& $(0,c,0,d;0,a,0,b)$ & \\\cline{2-4}
                                  & $U3UA3$& & \multirow{4}{*}{$P3_1$ (144)/$P3_2$ (145)}\\
                                  & $U4UA4$& & \\
                                  & $D1DA1$&  & \\
                                  & $D2DA2$&  & \\     
\hline
    \multirow{14}{*}{$P6cc$ (184)} & $DT5DU5$& $(0,-a,0,-b)$ & $P6_2$ (171)\\
                                  & $DT5DU5$& $(a,0,b,0)$& $P6_4$ (172)\\\cline{2-4}
                                  & $U1UA1$& & \multirow{5}{*}{$P6_2$ (171)/$P6_4$ (172)}\\\
                                  & $U2UA2$& & \\  
                                  & $P3PA3$&  &  \\
                                  & $D1DA1$&  &  \\
                                  & $D2DA2$&  &  \\\cline{2-4}
                                  & $DT6DU6$& $(a,0,b,0)$ & $P3_1$ (144)\\
                                  & $DT6DU6$& $(0,-a,0,-b)$ & $P3_2$ (145)\\\cline{2-4}
                                  & $U3UA3$& & \multirow{5}{*}{$P3_1$ (144)/$P3_2$ (145)}\\
                                  & $U4UA4$& & \\
                                  & $P3PA3$&  & \\
                                  & $D1DA1$&  & \\
                                  & $D2DA2$&  &\\ 
\hline
    \end{tabular}
    \caption{Table IV continued. }
\end{table*}

\begin{table*}
    \begin{tabular}{c|ccc}
    \hline
     high sym. achiral sg.   & IRREP & domain direction & enantiomorphic sg. \\
     \hline 
\hline
    \multirow{12}{*}{$P6_3cm$ (185)} & $DT5DU5$& $(a,0,b,0)$ & $P3_1$ (144)\\
                                    & $DT5DU5$& $(0,-a,0,-b)$ & $P3_2$ (145)\\\cline{2-4}
                                    & $U1UA1$& & \multirow{5}{*}{$P3_1$ (144)/$P3_2$ (145)}\\\
                                    & $U2UA2$& & \\
                                    & $P3PA3$ &  &\\
                                    & $D1DA1$&  & \\
                                    & $D2DA2$&  & \\\cline{2-4}
                                    & $DT6DU6$& $(a,0,b,0)$ & $P6_1$ (169)\\
                                    & $DT6DU6$& $(0,-a,0,-b)$ & $P6_5$ (170)\\\cline{2-4}
                                    & $U3UA3$& & \multirow{5}{*}{$P6_1$ (169)/$P6_5$ (170)}\\
                                    & $U4UA4$& & \\ 
                                    & $P3PA3$&  &  \\
                                    & $D1DA1$&  &  \\
                                    & $D2DA2$&  &  \\
\hline
    \multirow{15}{*}{$P6_3mc$ (186)} & $DT5DU5$& $(a,0,b,0)$ & \multirow{2}{*}{$P3_1$ (144)}\\
                                    & $P3PA3$ & $(a,0,b,0;c,0,d,0)$ & \\
                                    & $DT5DU5$& $(0,-\cos{(\frac{\pi}{3})}a+\sin{(\frac{\pi}{3})}b,0,-\sin{(\frac{\pi}{3})}a-\cos{(\frac{\pi}{3})}b $& \multirow{2}{*}{$P3_2$ (145)}\\ 
                                    & $P3PA3$ & $(0,c,0,d;0,a,0,b)$ & \\\cline{2-4}
                                    & $U1UA1$& & \multirow{4}{*}{$P3_1$ (144)/$P3_2$ (145)}\\\
                                    & $U2UA2$& & \\
                                    & $D1DA1$&  & \\
                                    & $D2DA2$&  & \\\cline{2-4}
                                    & $DT6DU6$& $(a,0,b,0)$ & $P6_1$ (169)\\
                                    & $DT6DU6$& $(0,-\cos{(\frac{\pi}{3})}a+\sin{(\frac{\pi}{3})}b,0,-\sin{(\frac{\pi}{3})}a-\cos{(\frac{\pi}{3})}b$ & $P6_5$ (170)\\\cline{2-4} 
                                    & $U3UA3$& & \multirow{5}{*}{$P6_1$ (169)/$P6_5$ (170)}\\
                                    & $U4UA4$& & \\
                                    & $P3PA3$&  &  \\
                                    & $D1DA1$&  &  \\
                                    & $D2DA2$&  &  \\
\hline
    \multirow{11}{*}{$P\bar{6}m2$ (187)}  &  $DT3$& $(0,\cos{(\frac{\pi}{3})}a+\sin{(\frac{\pi}{3})}b,0,\sin{(\frac{\pi}{3})}a-\cos{(\frac{\pi}{3})}b$ & $P3_1$ (144) \\
                                        &  $DT3$& $(a,0,b,0)$ & $P3_2$ (145) \\\cline{2-4}
                                      & $U1$&  & \multirow{4}{*}{$P3_1$ (144)/$P3_2$ (145)}\\
                                      & $U2$&  & \\
                                      & $D1DA1$&  & \\
                                      & $D2DA2$&  & \\\cline{2-4}
                                      &  $DT3$&  & \multirow{5}{*}{$P3_112$ (151)/$P3_212$ (153)} \\
                                      &$U1$&  &  \\
                                      & $U2$&  &  \\
                                      & $D1DA1$&  &  \\
                                      & $D2DA2$&  &  \\
\hline
    \end{tabular}
    \caption{Table IV continued. }
\end{table*}

\begin{table*}
    \begin{tabular}{c|ccc}
    \hline
     high sym. achiral sg.   & IRREP & domain direction & enantiomorphic sg. \\
     \hline
\hline
    \multirow{11}{*}{$P\bar{6}c2$ (188)} &  $DT3$& $(0,a,0,-b)$ & $P3_1$ (144) \\
                                        &  $DT3$& $(a,0,b,0)$ & $P3_2$ (145) \\\cline{2-4}
                                      & $U1$&  & \multirow{4}{*}{$P3_1$ (144)/$P3_2$ (145)}\\
                                      & $U2$&  & \\
                                      & $D1DA1$&  & \\
                                      & $D2DA2$&  & \\\cline{2-4}
                                      &$U1$&  & \multirow{5}{*}{$P3_112$ (151)/$P3_212$ (153)} \\
                                      &  $DT3$&  & \\
                                      & $U2$&  &  \\
                                      & $D1DA1$&  &  \\
                                      & $D2DA2$&  &  \\
\hline
    \multirow{8}{*}{$P\bar{6}2m$ (189)} & $DT3$ & $(0,\cos{(\frac{\pi}{3})}a+\sin{(\frac{\pi}{3})}b,0,\sin{(\frac{\pi}{3})}a-\cos{(\frac{\pi}{3})}b$ & $P3_1$ (144) \\
                                      & $DT3$ & $(a,0,b,0)$ & $P3_2$ (145) \\\cline{2-4}
                                      & $U1$&  & \multirow{3}{*}{$P3_1$ (144)/$P3_2$ (145)}\\
                                      & $U2$&  &  \\
                                      & $D1$&  &  \\\cline{2-4}  
                                      & $DT3$&  & \multirow{4}{*}{$P3_112$ (151)/$P3_212$ (153)} \\
                                      & $U1$&  &  \\
                                      & $U2$&  &  \\
                                      & $D1$&  &  \\
\hline
    \multirow{8}{*}{$P\bar{6}2c$ (190)} & $DT3$ & $(0,-a,0,-b)$ & $P3_1$ (144) \\
                                      & $DT3$ & $(a,0,b,0)$ & $P3_2$ (145) \\\cline{2-4}
                                      & $U1$&  & \multirow{3}{*}{$P3_1$ (144)/$P3_2$ (145)} \\
                                      & $U2$&  &  \\
                                      & $D1$&  & \\\cline{2-4}  
                                      & $DT3$&  &  \multirow{4}{*}{$P3_112$ (151)/$P3_212$ (153)} \\
                                      & $U1$&  &  \\
                                      & $U2$&  &  \\
                                      & $D1$&  &   \\                  

\hline
    \end{tabular}
    \caption{Table IV continued. }
\end{table*}

\begin{table*}
    \begin{tabular}{c|ccc}
    \hline
     high sym. achiral sg.   & IRREP & domain direction & enantiomorphic sg. \\
     \hline

\hline
    \multirow{19}{*}{$P6/mmm$ (191)}    & $DT6$ & $(0,-a,0,b)$ & \multirow{2}{*}{$P3_1$ (144)} \\
                                        & $P3$ & $(0,-a,0,b;0,-c,0,d)$ &  \\
    \multirow{17}{*}{$P6/mcc$ (192)}    & $DT6$ & $(a,0,b,0)$ & \multirow{2}{*}{$P3_2$ (145)} \\
                                        & $P3$ & $(a,0,b,0;c,0,d,0)$ &   \\\cline{2-4}
                                        & $U3, U4$ &  & \multirow{2}{*}{$P3_1$ (144)/$P3_2$ (145)} \\
                                        & $D1, D2$ &  &   \\\cline{2-4}
                                        & $DT6$&  & \multirow{3}{*}{$P3_112$ (151)/$P3_212$ (153)} \\
                                        & $U3, U4$&  &   \\
                                        & $D1, D2, P3$ &  & \\\cline{2-4}
                                        & $DT6$&  &\multirow{3}{*}{ $P3_121$ (152)/$P3_221$ (154)}\\
                                        & $U3, U4$&  &  \\
                                        & $D1, D2, P3$ &  & \\\cline{2-4}
                                        & $DT5$ & $(0,-a,0,b)$ & $P6_2$ (171) \\
                                        & $DT5$ & $(a,0,b,0)$ & $P6_4$ (172) \\\cline{2-4}
                                        & $U1, U2$ &  & \multirow{2}{*}{$P6_2$ (171)/$P6_4$ (172)} \\
                                        & $D1, D2, P3$ &  &   \\\cline{2-4}
                                        & $DT5$&  & \multirow{3}{*}{$P6_222$ (180)/$P6_422$ (181)} \\
                                        & $U1, U2$&  &   \\
                                        & $D1, D2, P3$&  &   \\
\hline
    \multirow{19}{*}{$P6_3/mcm$ (193)}  & $DT5$ & $(a,0,b,0)$ & \multirow{2}{*}{$P3_1$ (144)} \\
                                        & $P3$ & $(0,-a,0,b;0,-c,0,d)$ &  \\
    \multirow{17}{*}{$P6_3/mmc$ (194)}  & $DT5$ & $(0,-a,0,b)$ & \multirow{2}{*}{$P3_2$ (145)} \\
                                        & $P3$ & $(a,0,b,0;c,0,d,0)$ &   \\\cline{2-4}
                                        & $U1, U2$ &  & \multirow{2}{*}{$P3_1$ (144)/$P3_2$ (145)} \\
                                        & $D1, D2$ &  &  \\\cline{2-4}
                                        &  $DT5$  & &\multirow{3}{*}{$P3_112$ (151)/$P3_212$ (153)} \\
                                        & $U1, U2$  &   \\
                                        & $D1, D2, P3$ &  & \\\cline{2-4}
                                        & $DT5$&  & \multirow{3}{*}{ $P3_121$ (152)/$P3_221$ (154)}\\
                                        & $U1, U2$&  &  \\
                                        & $D1, D2, P3$ &  &  \\\cline{2-4}
                                        & $DT6$& $(0,-a,0,b)$ & $P6_1$ (169))\\
                                        & $DT6$& $(a,0,b,0)$ & $P6_5$ (170)\\\cline{2-4}
                                        & $U3, U4$&  & \multirow{2}{*}{$P6_2$ (171)/$P6_4$ (172)}\\
                                        & $D1, D2, P3$&  &  \\\cline{2-4}
                                        & $DT6$&  & \multirow{3}{*}{$P6_222$ (180)/$P6_422$ (181)} \\
                                        & $U3, U4$&  &   \\
                                        & $D1, D2, P3$&  &   \\
\hline
    \end{tabular}
    \caption{Table IV continued. }
\end{table*}

\clearpage
\subsection{Chiral phonons from Sohncke non-enantiomorphic space-groups}
We now apply the same analysis to non-enantiomorphic Sohncke space groups. As in the previous section, we identify the q-point of the chiral phonon and the domain direction leading to each enantiomorphic subgroup. When the number of possible domain directions is large, only representative ones are reported.


\begin{table*}[b]
\scriptsize
    \begin{tabular}{c|ccc}
    \hline
     high sym. sohncke sg.   & IRREP & domain direction & enantiomorphic sg. \\
     \hline 

\hline
    \multirow{2}{*}{$P4_2$ (77)}    & $Z3$& $(a)$ & $P4_1$ (76) \\
                                    & $Z4$& $(a)$ & $P4_3$ (78) \\
\hline
    \multirow{4}{*}{$I4_1$ (80)}    & $X1$& $(a;a), (a;-a)$ & $P4_1$ (76) \\
                                    & $X2$& $(a;a), (a;-a)$ & $P4_3$ (78) \\
                                    & $M3$& $(a)$ & $P4_1$ (76) \\
                                    & $M4$& $(a)$ & $P4_3$ (78) \\
\hline
    \multirow{4}{*}{$P4_222$ (93)}  & $Z1$& $(a)$ & $P4_122$ (91) \\
                                    & $Z2$& $(a)$ & $P4_322$ (95) \\
                                    & $Z3$& $(a)$ & $P4_122$ (91) \\
                                    & $Z4$& $(a)$ & $P4_322$ (95) \\
\hline
    \multirow{4}{*}{$P4_22_12$ (94)}    & $Z1$& $(a)$ & $P4_12_12$ (92) \\
                                        & $Z2$& $(a)$ & $P4_32_12$ (96) \\
                                        & $Z3$& $(a)$ & $P4_12_12$ (92) \\
                                        & $Z4$& $(a)$ & $P4_32_12$ (96) \\
\hline
    \multirow{8}{*}{$I4_122$ (98)}      & $X1$& $(a;a), (-a;a)$ & $P4_122$ (91) \\
                                        & $X3$& $(a;a), (-a;a)$ & $P4_322$ (95) \\
                                        & $X2$& $(a;a), (-a;a)$ & $P4_12_12$ (92) \\
                                        & $X4$& $(a;a), (-a;a)$ & $P4_32_12$ (96) \\
                                        & $M1$& $(a)$ & $P4_122$ (91) \\
                                        & $M2$& $(a)$ & $P4_322$ (95) \\
                                        & $M3$& $(a)$ & $P4_12_12$ (92) \\
                                        & $M4$& $(a)$ & $P4_32_12$ (96) \\
\hline
\hline
    \multirow{2}{*}{$P4$ (75)}  & $LD4LE4$& $(a,b)$ & $P4_1$ (76) \\
                                & $LD3LE3$& $(a,b)$ & $P4_3$ (78) \\
\hline
    \multirow{2}{*}{$I4$ (79)}  & $LD4LE4$& $(a;b)$ & $P4_1$ (76) \\
                                & $LD3LE3$& $(a;b)$ & $P4_3$ (78) \\
\hline
    \multirow{4}{*}{$P422$ (89)}    & $LD4$& $(a,b)$ & $P4_1$ (76) \\
                                    & $LD3$& $(a,b)$ & $P4_3$ (78) \\
                                    & $LD4$& $(a,0),(a,a)$ & $P4_122$ (91) \\
                                    & $LD3$& $(a,0),(a,a)$ & $P4_322$ (95) \\
\hline
    \multirow{4}{*}{$P42_12$ (90)}  & $LD4$& $(a,b)$ & $P4_1$ (76) \\
                                    & $LD3$& $(a,b)$ & $P4_3$ (78) \\
                                    & $LD4$& $(a,0),(a,a)$ & $P4_12_12$ (92) \\
                                    & $LD3$& $(a,0),(a,a)$ & $P4_32_12$ (96) \\
\hline
    \multirow{6}{*}{$I422$ (97)}    & $LD4$& $(a,b)$ & $P4_1$ (76) \\
                                    & $LD3$& $(a,b)$ & $P4_3$ (78) \\
                                    & $LD4$& $(a,0)$ & $P4_122$ (91) \\
                                    & $LD3$& $(a,0)$ & $P4_322$ (95) \\
                                    & $LD4$& $(a, \tan{(\frac{\pi}{8})}a)$ & $P4_12_12$ (92) \\
                                    & $LD3$& $(a, \tan{(\frac{\pi}{8})}a)$ & $P4_32_12$ (96) \\
\hline
    \end{tabular}
    \caption{Sohncke non-enantiomorphic tetragonal space groups (first column) that hold chiral phonons (second column) capable of generating an enantiomorphic subgroup (fourth column). The third column lists the domain direction that is associated with the enantiomorphic subgroup of column four.  }
\end{table*}

\begin{table*}
    \begin{tabular}{c|ccc}
    \hline
     high sym. sohncke sg.   & IRREP & domain direction & enantiomorphic sg. \\
     \hline 

\hline
    \multirow{3}{*}{$P3$ (143)} & $DT3DU3$ &  $(a,b)$  & $P3_1$ (144) \\
                                & $DT2DU2$& $(a,b)$ & $P3_2$ (145) \\
                                & $U1UA1$&  &  $P3_1$ (144)/$P3_2$ (145)\\
\hline
    \multirow{6}{*}{$R3$ (146)} & $LD3LE3$ &  $(a,b)$  & \multirow{3}{*}{$P3_1$ (144)} \\
                                & $SM1SN1$& $(a,b;-\cos{(\frac{\pi}{3})}a+\sin{(\frac{\pi}{3})}b,-\sin{(\frac{\pi}{3})}a-\cos{(\frac{\pi}{3})}b;a,b)$ &  \\
                                & $Y1YA1$& $(a,b;-\cos{(\frac{\pi}{3})}a-\sin{(\frac{\pi}{3})}b,\sin{(\frac{\pi}{3})}a-\cos{(\frac{\pi}{3})}b;a,b)$ &  \\
                                & $LD2LE2$& $(a,b)$ & \multirow{3}{*}{$P3_2$ (145)} \\
                                & $SM1SN1$&  $(a,b;-\cos{(\frac{\pi}{3})}a-\sin{(\frac{\pi}{3})}b,\sin{(\frac{\pi}{3})}a-\cos{(\frac{\pi}{3})}b;a,b)$  &  \\
                                & $Y1YA1$&  $(a,b;-\cos{(\frac{\pi}{3})}a+\sin{(\frac{\pi}{3})}b,-\sin{(\frac{\pi}{3})}a-\cos{(\frac{\pi}{3})}b;a,b)$  &  \\
\hline
    \multirow{9}{*}{$P312$ (149)}   & $DT3$& $(a,b)$ & $P3_1$ (144) \\
                                    & $DT2$& $(a,b)$ & $P3_2$ (145) \\
                                    & $U1$&  & \multirow{2}{*}{$P3_1$ (144)/$P3_2$ (145)}\\
                                    & $D1DA1$ & &  \\
                                    & $DT3$& $(a,0), (-\cos{(\frac{\pi}{3})}a,-\sin{(\frac{\pi}{3})}a), (-\cos{(\frac{\pi}{3})}a,\sin{(\frac{\pi}{3})}a)$ & \multirow{2}{*}{$P3_112$ (151)} \\
                                    & $U1$& $(a, \tan{(\frac{\pi}{3})}a; -a,\tan{(\frac{\pi}{3})}a;2a,0)$ &  \\
                                    & $DT2$& $(a,0), (-\cos{(\frac{\pi}{3})}a,-\sin{(\frac{\pi}{3})}a), (-\cos{(\frac{\pi}{3})}a,\sin{(\frac{\pi}{3})}a)$ & \multirow{2}{*}{$P3_212$ (153)} \\
                                    & $U1$& $(a,-\tan{(\frac{\pi}{3})}a;  a,\tan{(\frac{\pi}{3})}a;2a,0)$ &  \\
                                    & $D1DA1$ & &  $P3_112$ (151)/$P3_212$ (153)\\ \cline{2-4}
\hline
    \multirow{11}{*}{$P321$ (150)}  & $DT3$& $(a,b)$ & \multirow{2}{*}{$P3_1$ (144)} \\
                                    & $P3PA3$& $(a,b;c,d)$ &  \\
                                    & $DT2$& $(a,b)$ & \multirow{2}{*}{$P3_2$ (145)} \\
                                    & $P2PA2$& $(a,b;c,d)$ &  \\
                                    & $U1$&  &  $P3_1$ (144)/$P3_2$ (145)\\ \cline{2-4}
                                    & $P3PA3$& $(a,b;a,b)$ & $P3_112$ (151) \\
                                    & $P2PA2$& $(a,b;a,b)$ & $P3_212$ (153) \\ \cline{2-4}
                                    & $DT3$& $(a,0), (-\cos{(\frac{\pi}{3})}a,-\sin{(\frac{\pi}{3})}a), (-\cos{(\frac{\pi}{3})}a,\sin{(\frac{\pi}{3})}a)$ & \multirow{2}{*}{$P3_121$ (152)} \\
                                    & $U1$& $(a, \tan{(\frac{\pi}{3})}a; -a,\tan{(\frac{\pi}{3})}a;2a,0)$ &  \\
                                    & $DT2$& $(a,0), (-\cos{(\frac{\pi}{3})}a,-\sin{(\frac{\pi}{3})}a), (-\cos{(\frac{\pi}{3})}a,\sin{(\frac{\pi}{3})}a)$ & \multirow{2}{*}{$P3_221$ (154)} \\
                                    & $U1$& $(a,-\tan{(\frac{\pi}{3})}a;  a,\tan{(\frac{\pi}{3})}a;2a,0)$ & \\
\hline
    \end{tabular}
    \caption{Sohncke non-enantiomorphic trigonal and hexagonal space groups (first column) that hold chiral phonons (second column) capable of generating an enantiomorphic subgroup (fourth column). The third column lists the domain direction that is associated with the enantiomorphic subgroup of column four, they are not all listed for clarity but they all have domain directions.  }
\end{table*}

\begin{table*}
    \begin{tabular}{c|ccc}
    \hline
     high sym. sohncke sg.   & IRREP & domain direction & enantiomorphic sg. \\
     \hline 

\hline
    \multirow{14}{*}{$R32$ (155)}   & $LD3$& $(a,b)$ & \multirow{2}{*}{$P3_1$ (144)} \\
                                    & $SM2SN2$& $(a,b;-\cos{(\frac{\pi}{3})}a+\sin{(\frac{\pi}{3})}b,-\sin{(\frac{\pi}{3})}a-\cos{(\frac{\pi}{3})}b;a,b)$ & \\
                                    & $LD2$& $(a,b)$ & \multirow{2}{*}{$P3_2$ (145)} \\
                                    & $SM2SN2$& $(a,b;-\cos{(\frac{\pi}{3})}a-\sin{(\frac{\pi}{3})}b,\sin{(\frac{\pi}{3})}a-\cos{(\frac{\pi}{3})}b;a,b)$ & \\
                                    & $GP1GQ1$ & &  $P3_1$ (144)/$P3_2$ (145)\\ \cline{2-4}
                                    & $LD3$& $(a,0), (-\cos{(\frac{\pi}{3})}a,-\sin{(\frac{\pi}{3})}a), (-\cos{(\frac{\pi}{3})}a,\sin{(\frac{\pi}{3})}a)$ & $P3_121$ (152) \\
                                    & $LD2$& $(a,0), (-\cos{(\frac{\pi}{3})}a,-\sin{(\frac{\pi}{3})}a), (-\cos{(\frac{\pi}{3})}a,\sin{(\frac{\pi}{3})}a)$ & $P3_221$ (154) \\\cline{2-4}
                                    & $SM1SN1$& $(a,b;-\cos{(\frac{\pi}{3})}a+\sin{(\frac{\pi}{3})}b,-\sin{(\frac{\pi}{3})}a-\cos{(\frac{\pi}{3})}b;a,b)$ & \multirow{3}{*}{$P3_112$ (151)}\\
                                    & $Y1YA1$& $(a,b;-\cos{(\frac{\pi}{3})}a-\sin{(\frac{\pi}{3})}b,\sin{(\frac{\pi}{3})}a-\cos{(\frac{\pi}{3})}b;a,b)$&  \\
                                    & $Y2YA2$& $(a,b;-\cos{(\frac{\pi}{3})}a-\sin{(\frac{\pi}{3})}b,\sin{(\frac{\pi}{3})}a-\cos{(\frac{\pi}{3})}b;a,b)$&  \\
                                    & $SM1SN1$& $(a,b;-\cos{(\frac{\pi}{3})}a-\sin{(\frac{\pi}{3})}b,\sin{(\frac{\pi}{3})}a-\cos{(\frac{\pi}{3})}b;a,b)$ &\multirow{3}{*}{$P3_212$ (153)} \\
                                    & $Y1YA1$&$(a,b;-\cos{(\frac{\pi}{3})}a+\sin{(\frac{\pi}{3})}b,-\sin{(\frac{\pi}{3})}a-\cos{(\frac{\pi}{3})}b;a,b)$&  \\
                                    & $Y2YA2$&$(a,b;-\cos{(\frac{\pi}{3})}a+\sin{(\frac{\pi}{3})}b,-\sin{(\frac{\pi}{3})}a-\cos{(\frac{\pi}{3})}b;a,b)$&  \\
                                    & $GP1GQ1$ & &  $P3_112$ (151)/$P3_212$ (153)\\
\hline
    \multirow{12}{*}{$P6$ (168)} & $DT6DU6$& $(a,b) $ & \multirow{2}{*}{$P3_1$ (144)}\\
                                & $P3PA3$ &  $(a,b;c,d)$  &  \\
                                & $DT4DU4$ &  $(a,b)$  & \multirow{2}{*}{$P3_2$ (145)} \\
                                & $P2PA2$ &  $(a,b;c,d)$  &   \\
                                & $U2UA2$ &   &  \multirow{2}{*}{$P3_1$ (144)/$P3_2$ (145)} \\
                                & $D1DA1$ &   &  \\ \cline{2-4}
                                & $DT3DU3$ &  $(a,b)$  &  \multirow{2}{*}{$P6_2$ (171)} \\
                                & $P2PA2$ &  $(a,b;a,b)$  &  \\
                                & $DT5DU6$& $(a,b)$ &  \multirow{2}{*}{$P6_4$ (172)} \\
                                & $P3PA3$ &  $(a,b;a,b)$  & \\
                                & $U1UA1$ &   &  \multirow{2}{*}{$P6_2$ (171)/$P6_4$ (172)} \\
                                & $D1DA1$ &   &  \\
\hline
    \multirow{12}{*}{$P6_3$ (173)}  & $DT5DU5$& $(a,b) $ & \multirow{2}{*}{$P3_1$ (144)}\\
                                    & $P3PA3$& $(a,b;c,d) $ &  \\
                                    & $DT3DU3$& $(a,b) $ & \multirow{2}{*}{$P3_2$ (145)}\\
                                    & $P2PA2$& $(a,b;c,d) $ &  \\
                                    & $U1UA1$ &   &  \multirow{2}{*}{$P3_1$ (144)/$P3_2$ (145)} \\
                                    & $D1DA1$ &   &  \\ \cline{2-4}
                                    & $DT6DU6$& $(a,b) $ & \multirow{2}{*}{$P6_1$ (169)}\\
                                    & $P3PA3$&  $(a,b;-a,-b)$  & \\
                                    & $DT4DU4$&  $(a,b)$  & \multirow{2}{*}{$P6_5$ (170)} \\
                                    & $P2PA2$&  $(a,b;-a,-b)$  & \\
                                    & $U2UA2$ &   &  \multirow{2}{*}{$P6_1$ (169)/$P6_5$ (170)} \\
                                    & $D1DA1$ &   &  \\
\hline
    \end{tabular}
    \caption{Table XI continued. }
\end{table*}

\begin{table*}
    \begin{tabular}{c|ccc}
    \hline
     high sym. sohncke sg.   & IRREP & domain direction & enantiomorphic sg. \\
     \hline 

\hline
    \multirow{32}{*}{$P622$ (177)}  & $DT3$& $(a,0) $ & $P6_222$ (180)\\
                                    & $DT3$& $(a,b) $ & $P6_2$ (171)\\
                                    & $DT4$& $(a,0) $ & $P3_221$ (154)\\
                                    & $DT4$& $(a,\cot{(\frac{\pi}{3})}a) $ & $P3_212$ (153)\\
                                    & $DT4$& $(a,b) $ & $P3_2$ (145)\\
                                    & $DT5$& $(a,0) $ & $P6_422$ (181)\\
                                    & $DT5$& $(a,b) $ & $P6_4$ (172)\\
                                    & $DT6$& $(a,0) $ & $P3_121$ (152)\\
                                    & $DT6$& $(a,\cot{(\frac{\pi}{3})}a) $ & $P3_112$ (151)\\
                                    & $DT6$& $(a,b) $ & $P3_1$ (144)\\
                                    & $P2$& $(a,0;a,0) $ & $P6_222$ (180)\\
                                    & $P2$& $(a,b;a,b) $ & $P6_2$ (171)\\
                                    & $P2$& $(a,0;b,0) $ & $P3_221$ (154)\\
                                    & $P2$& $(a,b;a,-b) $ & $P3_212$ (153)\\
                                    & $P2$& $(a,b;c,d) $ & $P3_2$ (145)\\
                                    & $P3$& $(a,0;a,0) $ & $P6_422$ (181)\\
                                    & $P3$& $(a,b;a,b) $ & $P6_4$ (172)\\
                                    & $P3$& $(a,0;b,0) $ & $P3_121$ (152)\\
                                    & $P3$& $(a,b;a,-b) $ & $P3_112$ (151)\\
                                    & $P3$& $(a,b;c,d) $ & $P3_1$ (144)\\
                                    & $U1$& $(a, \tan{(\frac{\pi}{3})}a; -a,\tan{(\frac{\pi}{3})}a;2a,0)$ &  $P6_222$ (180)\\
                                    & $U1$& $(a,-\tan{(\frac{\pi}{3})}a;  a,\tan{(\frac{\pi}{3})}a;2a,0)$ &  $P6_422$ (181)\\
                                    & $U1$& & $P6_2$ (171)/$P6_4$ (172)\\
                                    & $U2$& $(a,-\tan{(\frac{\pi}{3})}a;  a,\tan{(\frac{\pi}{3})}a;2a,0)$ & $P3_112$ (151) \\
                                    & $U2$& $(a, \tan{(\frac{\pi}{3})}a; a,-\tan{(\frac{\pi}{3})}a;2a,0)$ & $P3_121$ (153) \\
                                    & $U2$&  & $P3_121$ (152)/$P3_212$ (154)\\
                                    & $U2$&  & $P3_1$ (144)/$P3_2$ (145)\\
                                    & $D1$&  & $P6_222$ (180)/$P6_422$ (181)\\
                                    & $D1$&  & $P6_2$ (171)/$P6_4$ (172)\\
                                    & $D1$&  & $P3_112$ (151)/$P3_212$ (153)\\
                                    & $D1$&  & $P3_121$ (152)/$P3_221$ (154)\\
                                    & $D1$&  & $P3_1$ (144)/$P3_2$ (145)\\
\hline
    \end{tabular}
    \caption{Table XI continued. }
\end{table*}

\begin{table*}
    \begin{tabular}{c|ccc}
    \hline
     high sym. sohncke sg.   & IRREP & domain direction & enantiomorphic sg. \\
     \hline 

\hline
    \multirow{31}{*}{$P6_322$ (182)}& $DT3$& $(a,0) $ & $P3_221$ (154)\\
                                    & $DT3$& $(a,\cot{(\frac{\pi}{3})}a) $ & $P3_212$ (153)\\
                                    & $DT3$& $(a,b) $ & $P3_2$ (145)\\
                                    & $DT4$& $(a,0) $ & $P6_522$ (179)\\
                                    & $DT4$& $(a,b) $ & $P6_5$ (170)\\
                                    & $DT5$& $(a,0) $ & $P3_121$ (152)\\
                                    & $DT5$& $(a,\cot{(\frac{\pi}{3})}a) $ & $P3_112$ (151)\\
                                    & $DT5$& $(a,b) $ & $P3_1$ (144)\\
                                    & $DT6$& $(a,0) $ & $P6_122$ (178)\\
                                    & $DT6$& $(a,b) $ & $P6_1$ (169)\\
                                    & $P2$& $(a,\cot{(\frac{\pi}{3})}a;a,-\cot{(\frac{\pi}{3})}a) $ & $P6_522$ (179)\\
                                    & $P2$& $(a,b;-a,-b) $ & $P6_5$ (170)\\
                                    & $P2$& $(a,\cot{(\frac{\pi}{3})}a;b,0) $ & $P3_221$ (154)\\
                                    & $P2$& $(a,b;a,-b) $ & $P3_212$ (153)\\
                                    & $P2$& $(a,b;c,d) $ & $P3_2$ (145)\\
                                    & $P3$& $(a,\cot{(\frac{\pi}{3})}a;a,-\cot{(\frac{\pi}{3})}a) $ & $P6_122$ (178)\\
                                    & $P3$& $(a,b;-a-,b) $ & $P6_1$ (169)\\
                                    & $P3$& $(a,\cot{(\frac{\pi}{3})}a;b,0);b,0) $ & $P3_121$ (152)\\
                                    & $P3$& $(a,b;a,-b) $ & $P3_112$ (151)\\
                                    & $P3$& $(a,b;c,d) $ & $P3_1$ (144)\\
                                    & $U1$& $(a,-\tan{(\frac{\pi}{3})}a;  a,\tan{(\frac{\pi}{3})}a;2a,0)$ & $P3_121$ (152) \\
                                    & $U1$& $(a, \tan{(\frac{\pi}{3})}a; a,-\tan{(\frac{\pi}{3})}a;2a,0)$ & $P3_221$ (154) \\
                                    & $U1$&  &  $P3_112$ (151)/$P3_212$ (153)\\
                                    & $U1$&  & $P3_1$ (144)/$P3_2$ (145)\\
                                    & $U2$& & $P6_122$ (178)/$P6_522$ (179)\\
                                    & $U2$& & $P6_1$ (168)/$P6_5$ (169)\\
                                    & $D1$&  & $P6_122$ (178)/$P6_522$ (179)\\
                                    & $D1$&  & $P6_1$ (169)/$P6_5$ (170)\\
                                    & $D1$&  & $P3_112$ (151)/$P3_212$ (153)\\
                                    & $D1$&  & $P3_121$ (152)/$P3_221$ (154)\\
                                    & $D1$&  & $P3_1$ (144)/$P3_2$ (145)\\
\hline
    \end{tabular}
    \caption{Table XI continued. }
\end{table*}


\begin{table*}
    \begin{tabular}{c|ccc}
    \hline
     high sym. sohncke sg.   & IRREP & domain direction & enantiomorphic sg. \\
     \hline 

\hline
    \multirow{12}{*}{$P4_232$ (208)} & $X1$& $(a;0;0),(0;a;0),(0;0;a)$ & $P4_122$ (91) \\
                                    & $X2$& $(a;0;0),(0;a;0),(0;0;a)$ & $P4_322$ (95) \\
                                    & $X1$& $(a;b;a),(a;a;b),(b;a;a)$ & $P4_12_12$ (92) \\
                                    & $X2$& $(a;b;a),(a;a;b),(b;a;a)$ & $P4_32_12$ (96) \\
                                    & $X1$& $(a;a;a),(a;a;-a),(a;-a;a),(-a;-a;-a)$ & $P4_132$ (213) \\
                                    & $X2$& $(a;a;a),(a;a;-a),(a;-a;a),(-a;-a;-a)$ & $P4_332$ (212) \\ \cline{2-4}
                                    & $X3$& $(a;0;0),(0;a;0),(0;0;a)$ & $P4_122$ (91) \\
                                    & $X4$& $(a;0;0),(0;a;0),(0;0;a)$ & $P4_322$ (95) \\
                                    & $X3$& $(a;b;a),(a;a;b),(b;a;a)$ & $P4_12_12$ (92) \\
                                    & $X4$& $(a;b;a),(a;a;b),(b;a;a)$ & $P4_32_12$ (96) \\
                                    & $X3$& $(a;a;a),(a;a;-a),(a;-a;a),(-a;-a;-a)$ & $P4_132$ (213) \\
                                    & $X4$& $(a;a;a),(a;a;-a),(a;-a;a),(-a;-a;-a)$ & $P4_332$ (212) \\
\hline
    \multirow{17}{*}{$F4_132$ (210)} & $X1$& $(a;0;0)$ & $P4_12_12$ (92) \\
                                    & $X4$& $(a;0;0)$ & $P4_32_12$ (96) \\
                                    & $X4$& $(a;a;0)$ & $P4_12_12$ (92) \\
                                    & $X1$& $(a;a;0)$ & $P4_32_12$ (96) \\
                                    & $X3$& $(a;0;0)$ & $P4_122$ (91) \\
                                    & $X2$& $(a;0;0)$ & $P4_322$ (95) \\
                                    & $X3$& $(a;b;a)$ & $P4_12_12$ (92) \\
                                    & $X2$& $(a;b;a)$ & $P4_32_12$ (96) \\
                                    & $X3$& $(a;a;a)$ & $P4_132$ (213) \\
                                    & $X2$& $(a;a;a)$ & $P4_332$ (212) \\
                                    & $X5$& $(a;0;0,a;0,0)$ & $P4_122$ (91) \\
                                    & $X5$& $(a,0;0,0;a,0)$ & $P4_322$ (95) \\
                                    & $W1$&  & $P4_1$ (76)/$P4_3$  (78)\\
                                    & $W1$& $(a,0,0,0;0,0,0,0;0,a,0,0)$ & $P4_122$ (91) \\
                                    & $W1$& $(a,0,0,0;0,a,0,0;0,0,0,0)$ & $P4_322$ (95) \\
                                    & $W1$& $(a,0,a,0;0,a,0,a;0,0,0,0)$ & $P4_12_12$ (92)\\
                                    & $W1$& $(a,0,a,0;0,0,0,0;0,a,0,a)$ & $P4_32_12$ (96)\\
\hline
    \multirow{10}{*}{$I4_132$ (214)} & $H1$& $(a)$ & $P4_132$ (213) \\
                                    & $H2$& $(a)$ & $P4_332$ (212) \\
                                    & $H3$& $(a,0)$ & $P4_12_12$ (92) \\
                                    & $H3$& $(a,\cot{(\frac{\pi}{3})}a)$ & $P4_32_12$ (96) \\
                                    & $H4$& $(a;0;0)$ & $P4_122$ (91) \\
                                    & $H5$& $(a;0;0)$ & $P4_322$ (95) \\
                                    & $N1$& $(a;a;0;0;0;0)$ & $P4_122$ (91) \\
                                    & $N3$& $(a;a;0;0;0;0)$ & $P4_322$ (95) \\
                                    & $N2$& $(a;a;0;0;0;0)$ & $P4_12_12$ (92) \\
                                    & $N4$& $(a;a;0;0;0;0)$ & $P4_32_12$ (96) \\
\hline
    \end{tabular}
    \caption{Sohncke non-enantiomorphic cubic space groups (first column) that hold chiral phonons (second column) capable of generating an enantiomorphic subgroup (fourth column). The third column lists the domain direction that is associated with the enantiomorphic subgroup of column four.  }
\end{table*}

\begin{table*}
    \begin{tabular}{c|ccc}
    \hline
     high sym. sohncke sg.   & IRREP & domain direction & enantiomorphic sg. \\
     \hline 

\hline
    \multirow{4}{*}{$P23$ (195)}    & $LD3LE3$& $(a,b;0,0;0,0;0,0)$ & $P3_1$ (144) \\
                                    & $LD2LE2$& $(a,b;0,0;0,0;0,0)$ & $P3_2$ (145) \\
                                    & $SM1$&  & $P3_1$ (144)/$P3_2$ (145) \\
                                    & $S1$&  & $P3_1$ (144)/$P3_2$ (145) \\
\hline
    \multirow{4}{*}{$F23$ (196)}    & $LD3LE3$& $(a,b;0,0;0,0;0,0)$ & $P3_1$ (144) \\
                                    & $LD2LE2$& $(a,b;0,0;0,0;0,0)$ & $P3_2$ (145) \\
                                    & $SM1$&  & $P3_1$ (144)/$P3_2$ (145) \\
                                    & $Q1QA1$&  & $P3_1$ (144)/$P3_2$ (145) \\
\hline
    \multirow{3}{*}{$I23$ (197)}    & $LD3LE3$& $(a,b;0,0;0,0;0,0)$ & $P3_1$ (144) \\
    \multirow{3}{*}{$I2_13$ (199)}  & $LD2LE2$& $(a,b;0,0;0,0;0,0)$ & $P3_2$ (145) \\
                                    & $SM1$&  & $P3_1$ (144)/$P3_2$ (145) \\
                                    & $G1$&  & $P3_1$ (144)/$P3_2$ (145) \\
\hline
    \multirow{4}{*}{$P2_13$ (198)}  & $LD3LE3$& $(a,b;0,0;0,0;0,0)$ & $P3_1$ (144) \\
                                    & $LD2LE2$& $(a,b;0,0;0,0;0,0)$ & $3_2$ (145) \\
                                    & $SM1$&  & $P3_1$ (144)/$P3_2$ (145) \\
                                    & $S1S1$&  & $P3_1$ (144)/$P3_2$ (145) \\
\hline
    \multirow{20}{*}{$P432$ (207)}  & $LD3$& $(a,b;0,0;0,0;0,0)$ & $P3_1$ (144) \\
  \multirow{20}{*}{\textit{Part 1}} & $LD2$& $(a,b;0,0;0,0;0,0)$ & $P3_2$ (145) \\
                                    & $LD3$& $(a,0;0,0;0,0;0,0)$ & $P3_121$ (152) \\
                                    & $LD2$& $(a,0;0,0;0,0;0,0)$ & $P3_221$ (154) \\
                                    & $SM1$&  & $P3_112$ (151)/$P3_212$ (153) \\
                                    & $SM2$&  & $P3_1$ (144)/$P4_3$ (145) \\ \cline{2-4}
                                    & $DT4$ & $(a,0;a,0;a,0),(a,a;a,a;a,a)$ & $P4_132$ (213) \\
                                    & $DT3$ & $(a,0;a,0;a,0),(a,a;a,a;a,a)$ & $P4_332$ (212) \\
                                    & $DT4$ & $(a,0;0,0;0,0),(a,a;0,0;0,0)$ & $P4_122$ (91) \\
                                    & $DT3$ & $(a,0;0,0;0,0),(a,a;0,0;0,0)$ & $P4_322$ (95) \\
                                    & $DT4$ & $(a,a;a,a;b,b),(a,a;a,a;b,0)$ & $P4_12_12$ (92) \\
                                    & $DT3$ & $(a,a;a,a;b,b),(a,a;a,a;b,0)$ & $P4_32_12$ (96) \\
                                    & $DT4$ & $(a,a;a,a;b,c),(a,b;c,0;c,0)$ & $P4_1$ (76) \\
                                    & $DT3$ & $(a,a;a,a;b,c),(a,b;c,0;c,0)$ & $P4_3$ (78) \\ \cline{2-4}
                                    & $Z1$ & $(a,a;a,a;a,a;-a,a;-a,a;-a,a)$ & $P4_132$ (213) \\
                                    & $Z1$ & $(a,-a;a,-a;a,a;a,-a;a,a;a,a)$ & $P4_332$ (212) \\
                                    & $Z1$ & $(a,a;0,0;0,b;b,0;0,0;a,a)$ & $P4_122$ (91) \\
                                    & $Z1$ & $(a,0;b,b;0,0;b,b;0,a;0,0)$ & $P4_322$ (95) \\
                                    & $Z1$ & $(a,a;0,b;0,0;a,a;0,0;b,0)$ & $P4_12_12$ (92) \\
                                    & $Z1$ & $(a,a;b,0;0,0;a,a;0,0;0,b)$ & $P4_32_12$ (96) \\
                                    & $Z1$ & $(a,-a;b,c;d,d;a,-a;d,d;c,-b)$ & $P4_1$ (76) \\
                                    & $Z1$ & $(a,-a;b,c;d,d;a,-a;d,d;-c,b)$ & $P4_3$ (78) \\ \cline{2-4}
\hline
    \end{tabular}
    \caption{Sohncke non-enantiomorphic cubic space groups (first column) that hold chiral phonons (second column) capable of generating an enantiomorphic subgroup (fourth column) on high-symmetry lines. The third column lists the domain direction that is associated with the enantiomorphic subgroup of column four.  }
\end{table*}

\begin{table*}
    \begin{tabular}{c|ccc}
    \hline
     high sym. sohncke sg.   & IRREP & domain direction & enantiomorphic sg. \\
     \hline 

\hline
    \multirow{22}{*}{$P432$ (207)}  & $Z2$ & $(a,0;a,0;a,0;0,-a;0,a;0,a)$ & $P4_132$ (213) \\
  \multirow{22}{*}{\textit{Part 2}} & $Z2$ & $(a,0;a,0;0,-a;a,0;0,a;0,a)$ & $P4_332$ (212) \\
                                    & $Z2$ & $(a,a;b,0;0,0;b,0;a,-a;0,0)$ & $P4_122$ (91) \\
                                    & $Z2$ & $(a,-a;b,0;0,0;b,0;a,a;0,0)$ & $P4_322$ (95) \\
                                    & $Z2$ & $(a,0;b,b;0,0;a,0;0,0;b,-b)$ & $P4_12_12$ (92) \\
                                    & $Z2$ & $(a,0;b,-b;0,0;a,0;0,0;b,b)$ & $P4_32_12$ (96) \\
                                    & $Z2$ & $(a,0;b,0;c,d;d,-c;0,b;0,a)$ & $P4_1$ (76) \\
                                    & $Z2$ & $(a,0;b,0;c,d;-d,c;0,b;0,a)$ & $P4_3$ (78) \\ \cline{2-4}
                                    & $T4$ & $(a,a;a,a;a,a)$ & $P4_132$ (213) \\
                                    & $T3$ & $(a,a;a,a;a,a)$ & $P4_332$ (212) \\
                                    & $T3$ & $(a,0;a,0;a,0)$ & $P4_132$ (213) \\
                                    & $T4$ & $(a,0;a,0;a,0)$ & $P4_332$ (212) \\
                                    & $T3$ & $(a,a;b,0;b,0)$ & $P4_122$ (91) \\
                                    & $T4$ & $(a,a;b,0;b,0)$ & $P4_322$ (95) \\
                                    & $T4$ & $(a,a;a,a;b,0)$ & $P4_122$ (91) \\
                                    & $T3$ & $(a,a;a,a;b,0)$ & $P4_322$ (95) \\
                                    & $T3$ & $(a,0;b,0;b,0)$ & $P4_12_12$ (92) \\
                                    & $T4$ & $(a,0;b,0;b,0)$ & $P4_32_12$ (96) \\
                                    & $T4$ & $(a,a;b,b;b,b)$ & $P4_12_12$ (92) \\
                                    & $T3$ & $(a,a;b,b;b,b)$ & $P4_32_12$ (96) \\
                                    & $T3$ & $(a,b;c,0;c,0)$ & $P4_1$ (76) \\
                                    & $T4$ & $(a,b;c,0;c,0)$ & $P4_3$ (78) \\
                                    & $T4$ & $(a,b;c,c;c,c)$ & $P4_1$ (76) \\
                                    & $T3$ & $(a,b;c,c;c,c)$ & $P4_3$ (78) \\
\hline
    \multirow{13}{*}{$F432$ (209)}  & $LD3$& $(a,b;0,0;0,0;0,0)$ & $P3_1$ (144) \\
\multirow{13}{*}{\textit{Part 1}}   & $LD2$& $(a,b;0,0;0,0;0,0)$ & $P4_3$ (145) \\
                                    & $LD3$& $(a,0;0,0;0,0;0,0)$ & $P3_121$ (152) \\
                                    & $LD2$& $(a,0;0,0;0,0;0,0)$ & $P3_221$ (154) \\
                                    & $SM1$&  & $P3_112$ (151)/$P3_212$ (153) \\
                                    & $SM2$&  & $P3_1$ (144)/$P4_3$ (145) \\ \cline{2-4}
                                    & $DT4$ & $(a,0;a,0;a,0),(a,a;a,a;a,a)$ & $P4_132$ (213) \\
                                    & $DT3$ & $(a,0;a,0;a,0),(a,a;a,a;a,a)$ & $P4_332$ (212) \\
                                    & $DT4$ & $(a,0;0,0;0,0)$ & $P4_122$ (91) \\
                                    & $DT3$ & $(a,0;0,0;0,0)$ & $P4_322$ (95) \\
                                    & $DT4$ & $(a,0;a,0;b,0),(a,a;b,0;b,0)$ & $P4_12_12$ (92) \\
                                    & $DT3$ & $(a,0;a,0;b,0),(a,a;b,0;b,0)$ & $P4_32_12$ (96) \\
                                    & $DT4$ & $(a,b;0,0;0,0),(a,b;c,0;c,0)$ & $P4_1$ (76) \\
                                    & $DT3$ & $(a,b;0,0;0,0),(a,b;c,0;c,0)$ & $P4_3$ (78) \\ 
\hline
    \end{tabular}
    \caption{Table XVI continued.  }
\end{table*}

\begin{table*}
    \begin{tabular}{c|ccc}
    \hline
     high sym. sohncke sg.   & IRREP & domain direction & enantiomorphic sg. \\
     \hline 

\hline
    \multirow{12}{*}{$F432$ (209)}  & $Q1QA1,Q2QA2$&  & $P4_132$ (213)/$P4_332$ (212) \\
\multirow{12}{*}{\textit{Part 2}}   & $Q1QA1,Q2QA2$&  & $P3_112$ (151)/$P3_212$ (153) \\
                                    & $Q1QA1,Q2QA2$&  & $P4_122$ (91)/$P4_322$ (95) \\
                                    & $Q1QA1,Q2QA2$&  & $P4_12_12$ (92)/$P4_32_12$ (96) \\
                                    & $Q1QA1,Q2QA2$&  & $P4_1$ (76)/$P4_3$ (78) \\  \cline{2-4}
                                    & $V2$ & $(a,0;a,0;0,a;a,0;0,a;0,a)$ & $P4_132$ (213) \\
                                    & $V2$ & $(a,0;0,a;a,0;0,a;0,-a;a,0)$ & $P4_332$ (212) \\ 
                                    & $V2$ & $(a,0;0,b;0,c;a,0;0,c;b,0)$ & $P4_12_12$ (92) \\
                                    & $V2$ & $(a,0;b,0;0,c;a,0;0,c;0,b)$ & $P4_32_12$ (96) \\ 
                                    & $V2$ & $(a,b;c,0;0,d;c,0;b,-a;0,d)$ & $P4_1$ (76) \\
                                    & $V2$ & $(a,b;c,0;0,d;c,0;-b,a;0,d)$ & $P4_3$ (78) \\ \cline{2-4} 
                                    & $C1$&  & $P3_121$ (152)/$P3_221$ (154) \\
                                    & $C1$&  & $P4_1$ (76)/$P4_3$ (78) \\  \cline{2-4} 
\hline
    \multirow{23}{*}{$I432$ (211)}  & $LD3$& $(a,b;0,0;0,0;0,0)$ & $P3_1$ (144) \\
                                    & $LD2$& $(a,b;0,0;0,0;0,0)$ & $P4_3$ (145) \\
                                    & $LD3$& $(a,0;0,0;0,0;0,0)$ & $P3_121$ (152) \\
                                    & $LD2$& $(a,0;0,0;0,0;0,0)$ & $P3_221$ (154) \\
                                    & $SM1$&  & $P3_112$ (151)/$P3_212$ (153) \\
                                    & $SM2$&  & $P3_1$ (144)/$P4_3$ (145) \\ \cline{2-4}
                                    & $DT4$ & $(a,0;a,0;a,0)$ & $P4_132$ (213) \\
                                    & $DT3$ & $(a,0;a,0;a,0)$ & $P4_332$ (212) \\
                                    & $DT4$ & $(a,0;0,0;0,0)$ & $P4_122$ (91) \\
                                    & $DT3$ & $(a,0;0,0;0,0)$ & $P4_322$ (95) \\
                                    & $DT4$ & $(a,0;a,0;b,0),(a,a;b,0;b,0)$ & $P4_12_12$ (92) \\
                                    & $DT3$ & $(a,0;a,0;b,0),(a,a;b,0;b,0)$ & $P4_32_12$ (96) \\
                                    & $DT4$ & $(a,b;0,0;0,0),(a,b;c,0;c,0)$ & $P4_1$ (76) \\
                                    & $DT3$ & $(a,b;0,0;0,0),(a,b;c,0;c,0)$ & $P4_3$ (78) \\ \cline{2-4}
                                    & $D2$ & $(a,0;0,a;0,a;a,0;a,0;0,a)$ & $P4_132$ (213) \\
                                    & $D2$ & $(a,0;0,-a;a,0;0,a;0,a;a,0)$ & $P4_332$ (213) \\
                                    & $D2$ & $(a,b;-a,b;b,a;b,-a;0,c;c,0)$ & $P4_122$ (91) \\
                                    & $D2$ & $(a,b;-a,b;b,a;b,-a;c,0;0,c)$ & $P4_322$ (95) \\
                                    & $D2$ & $(a,b;b,a;b,a;a,b;c,0;0,c)$ & $P4_12_12$ (92) \\
                                    & $D2$ & $(a,b;b,a;b,a;a,b;0,c;c,0)$ & $P4_32_12$ (96) \\
                                    & $D2$ & $(a,b;-b,a;c,d;d,c;d,c;c,d)$ & $P4_1$ (76) \\
                                    & $D2$ & $(a,b;b,-a;c,d;d,c;d,c;c,d)$ & $P4_3$ (78) \\ \cline{2-4}
                                    & $G1, G2$&  & $P3_112$ (151)/$P3_212$ (153) \\
\hline
    \end{tabular}
    \caption{Table XVI continued.  }
\end{table*}

\clearpage
The construction of the previous tables was made through the use of the utilities available on the Bilbao Crystallographic Server~\cite{Aroyo-06,Aroyo-06.2} and the ISOTROPY software suite~\cite{Hatch-03}. 
\end{document}